\documentclass[3p]{elsarticle}
 \makeatletter
 \def\ps@pprintTitle{%
  \def\@oddhead{\footnotesize \centerline{\textcopyright 2019 This manuscript version is made available under the CC-BY-NC-ND 4.0 license \url{http://creativecommons.org/licenses/by-nc-nd/4.0/}}}
  \let\@evenhead\@empty
  \def\@oddfoot{\footnotesize DOI:\text{ }\url{https://doi.org/10.1016/j.conengprac.2019.01.017}}%
 \let\@evenfoot\@oddfoot}
\makeatother
\usepackage{hyperref}
\hypersetup{
    colorlinks=true,
    linkcolor=blue,
    filecolor=magenta,      
    urlcolor=cyan,
}
\usepackage[T1]{fontenc}		
\usepackage[utf8]{inputenc}
\usepackage{lmodern}
\usepackage{color}
\usepackage{graphicx}
\usepackage{float}
\usepackage{amsthm}
\usepackage{amsmath}
\usepackage{subcaption}
\usepackage{amssymb}
\usepackage{color}
\usepackage{array}
\usepackage[export]{adjustbox}
\usepackage{cleveref}
\usepackage{enumerate}
\usepackage{mathtools}
\usepackage[ruled,lined]{algorithm2e}

\newtheorem{theorem}{Theorem}[section]
\newtheorem{lemma}{Lemma}[section]
\newtheorem*{remark}{Remark}

\usepackage{lineno} %count lines 
\modulolinenumbers[5]
\journal{Control Engineering Practice}
\begin{document}

\begin{frontmatter}

\title{Robust path-following control for articulated heavy-duty vehicles}

%\author{Filipe Marques Barbosa \fnref{myfootnote}, Lucas Barbosa Marcos\fnref{myfootnote}, Maira Martins\fnref{myfootnote} }

\author[electrical]{Filipe Marques Barbosa}
\ead{marquesfilipeb@usp.br}
\author[electrical]{Lucas Barbosa Marcos}
\ead{lucasbmarcos@usp.br}
\author[mechanical]{Maíra Martins da Silva}
\ead{mairams@sc.usp.br}
\author[electrical]{Marco Henrique Terra}
\ead{terra@sc.usp.br}
\author[electrical]{Valdir Grassi Junior\corref{correspondingauthor}}
\ead{vgrassi@usp.br}
\cortext[correspondingauthor]{Correspondence to: Department of Electrical and Computer Engineering, São Carlos School of Engineering, University of São Paulo, Av. Trabalhador São-carlense 400, 13566-590, São Carlos, SP, Brazil}

\address[electrical]{Department of Electrical and Computer Engineering, São Carlos School of Engineering, University of São Paulo, São Carlos, Brazil}
\address[mechanical]{Department of Mechanical Engineering, São Carlos School of Engineering, University of São Paulo, São Carlos, Brazil}
\begin{abstract}
Path following and lateral stability are crucial issues for autonomous vehicles. Moreover, these problems increase in complexity when handling articulated heavy-duty vehicles due to their poor manoeuvrability, large sizes and mass variation. In addition, uncertainties on mass may have the potential to significantly decrease the performance of the system, even to the point of destabilising it. These parametric variations must be taken into account during the design of the controller. However, robust control techniques usually require offline adjustment of auxiliary tuning parameters, which is not practical, leading to sub-optimal operation. Hence, this paper presents an approach to path-following and lateral control for autonomous articulated heavy-duty vehicles subject to parametric uncertainties by using a robust recursive regulator. The main advantage of the proposed controller is that it does not depend on the offline adjustment of tuning parameters. Parametric uncertainties were assumed to be on the payload, and an $\mathcal{H}_{\infty}$ controller was used for performance comparison. The performance of both controllers is evaluated in a double lane-change manoeuvre. Simulation results showed that the proposed method had better performance in terms of robustness, lateral stability, driving smoothness and safety, which demonstrates that it is a very promising control technique for practical applications.
\end{abstract}

\begin{keyword}
articulated vehicle; path following; lateral control; robust control; heavy-duty vehicle
\end{keyword}

\end{frontmatter}

\section{Introduction}

The advantages of autonomous vehicles are well-established in the academic literature. For example, reducing the number of accidents; easing the transportation of elderly and disabled people \cite{wu_2015}; offering more profitable means of transportation to industries and more efficient transportation methods to the military \cite{fu_2015,Shin2015}; improving ride comfort for passengers \cite{masserafilho_2017}; increasing road utilisation \cite{dias_2015}, etc.

Nowadays, heavy load vehicles are responsible for much of cargo transportation. The use of articulated heavy vehicles has been increasing due to their economic advantages \cite{kati2016robust}, freight transportation efficiency \cite{jujnovich2013path} and the growing demand for high capacity transport vehicles \cite{islam2015improve}. Furthermore, the same technologies used for autonomous cars can also be addressed to articulated heavy-duty vehicles \cite{fagnant2015preparing}, additionally increasing productivity and reducing cargo transportation costs \cite{noorvand2017autonomous}.

In the literature, different control techniques have been used to solve the path-following problem for autonomous vehicles. Alcala \textit{et al.} \cite{Alcala2018} used a Lyapunov-based technique with linear quadratic regulator - linear matrix inequality (LQR-LMI) tuning to solve the problem of guidance in an autonomous vehicle. Ji \textit{et al.} \cite{Ji2018} proposed a robust steering controller based on a backstepping variable structured control to maintain the yaw stability and minimise the lateral error. Mitraji \textit{et al.} \cite{Matraji2018} designed and implemented an adaptive Second Order Sliding Mode Control for a four wheels Skid-Steered Mobile Robot. The objective was to follow a predefined trajectory in the presence of disturbance and parametric uncertainties. Chu \textit{et al.} \cite{Chu2018} applied an active disturbance rejection control to a steering controller design with the aim to guarantee the lane keeping of the vehicle in the presence of uncertainties and external disturbance. Lastly, Hu \textit{et al.} \cite{Hu2016} presented an $\mathcal{H}_{\infty}$ output-feedback control strategy based on the mixed genetic algorithms and linear matrix inequality to perform the path following of autonomous ground vehicles. 

In addition, some authors have proposed the use of an active trailer steering system to improve path following and attitude control of articulated vehicles \cite{jujnovich2013path,Kim2016,Guan2017}. For instance, different vehicle conditions have been considered by Guan \textit{et al.} \cite{Guan2017} for deriving a model predictive control strategy. Regarding autonomous articulated vehicles, some control design strategies have been exploited in the literature. Yuan \textit{et al.} \cite{Yuan2016} proposed a lateral-longitudinal control scheme using automatic steering strategies to avoid jackknifing, considering input limitations. Micha{\l}ek \cite{Michaek2014} presented a highly scalable nonlinear cascade-like control to solve the path-following problem for articulated robotic vehicles equipped with arbitrary number of off-axle hitched trailers. With respect to the path-following problem for articulated vehicles, an active steering controller of the tractor and trailer based on LQR was designed by Kim \textit{et al.} \cite{Kim2016}, whilst a novel sliding mode controller was proposed by Nayl \textit{et al.} \cite{Nayl2018}. However, the autonomous control of articulated heavy-duty vehicles remains an issue. As payload may be much greater than vehicle weight itself \cite{kati2016robust}, mass is a critical parameter in vehicle dynamics and those vehicles are especially affected by mass variations. Hence, a control technique that overcomes the parametric uncertainties in the vehicle model is necessary, and it ensures system stability and performance objectives for a range of parameter values \cite{kati2016robust}. This leads to the need of robust controllers designed to withstand mass variations. 

Kati \textit{et al.} \cite{kati2016robust} proposed an $\mathcal{H}_{\infty}$ controller to deal with uncertainties on payload of the vehicle. However, as the $\mathcal{H}_{\infty}$ controller depends on the offline adjustment of the auxiliary parameter $\gamma$, this results in sub-optimal controller operation due to the mass variations. The $\mathcal{H}_{\infty}$ controller is furthermore robust, but it cannot ensure smoothness for steering control applications. In fact, the lower the $\gamma$ value, the more optimality condition  the controller reaches. On the other hand, it cannot guarantee driving smoothness as there is no parameter to deal with this. Consequently, a mixed $\mathcal{H}_{2}$/$\mathcal{H}_{\infty}$ controller is used in the literature, where smoothness, and robustness and optimisation are respectively handled \cite{Scherer1995}.  In order to address the sub-optimality problem, the contribution of this paper is a novel approach for the lateral control of an autonomous articulated heavy-duty vehicle, based on a Robust Linear Quadratic Regulator (RLQR) presented in \cite{Cerri2014} and \cite{cerri2009recursive}. The main advantage of the proposed controller is that it does not require any auxiliary tuning parameters, since both smoothness and robustness are already foreseen through a certain penalty parameter $\mu$, which vanishes in the limit when it tends to infinity. This feature  maintains the optimality for the full range of parametric uncertainties. This is additionally useful for online applications. A continuous-time model for the articulated vehicle in state-space form is presented. Then, the model is discretised in order to apply discrete RLQ control for solving a path-following problem.

Since $\mathcal{H}_\infty$ control is widely used for path-tracking problems \cite{Hu2016} and for robustifying the control strategy in automotive applications \cite{LI2018,ZHAO2018}, a standard $\mathcal{H}_\infty$ controller is also applied to the same plant for the sake of comparison. Uncertainties on vehicle mass are introduced, then the performance of both controllers is compared in different cases. Simulation tests evaluate robustness, steering behaviour, truck displacement error and orientation error.

The RLQR ensures stability for a range of possible payloads. On the other hand, the $\mathcal{H}_\infty$ controller is dependent on the auxiliary parameter $\gamma$. Therefore, it cannot maintain good performance (or even stability) for a wide range of payloads, unless $\gamma$ is adjusted offline \cite{Cerri2014}.

The paper is organised as follows: \autoref{system_modeling} presents both the model of a heavy articulated vehicle in continuous-time state-space form and a path-following model, which are properly put together to make a single model; \autoref{robust_recursive_regulator} exhibits the RLQR, showing how it is derived from a quadratic cost function and a robust regularised least squares problem; \autoref{application_and_results} shows and discusses the application of the RLQR and its results compared to an $\mathcal{H}_\infty$ controller; \autoref{conclusions} brings the conclusions.

\section{System modelling}
\label{system_modeling}

With the aim to make the articulated heavy-duty vehicle follow a desired path, it is not only necessary to minimise the lateral offset and heading error, but also ensure the vehicle stability. Therefore, the system modelling must take into account the path following and dynamic variables. This section introduces the vehicle model for simulations and control design.

\subsection{Path-following model}
\label{path-following-model}

In order to solve the path-following problem, the lateral controller aims to reduce lateral displacement and orientation angle errors of the towing vehicle. Therefore, the path-following model adopted here is based on the equations presented by Skjetne and Fossen \cite{skjetne2001nonlinear}. \autoref{path_following} shows the schematic diagram of path-following model for an articulated vehicle, where $\dot{y}_{1}$ is tractor lateral velocity and $v$ is tractor longitudinal velocity. The lateral displacement of the vehicle to a given reference path is the distance $\rho$ from tractor centre of gravity to the closest point $D$ on the desired path. The tractor orientation error is defined as $\theta = \psi_{1}-\psi_{des}$, where $\psi_{1}$ and $\psi_{des}$ are the current and desired orientation angles of the tractor, respectively.
%--
\begin{figure}[h]
      \centering
      \includegraphics[scale=0.5]{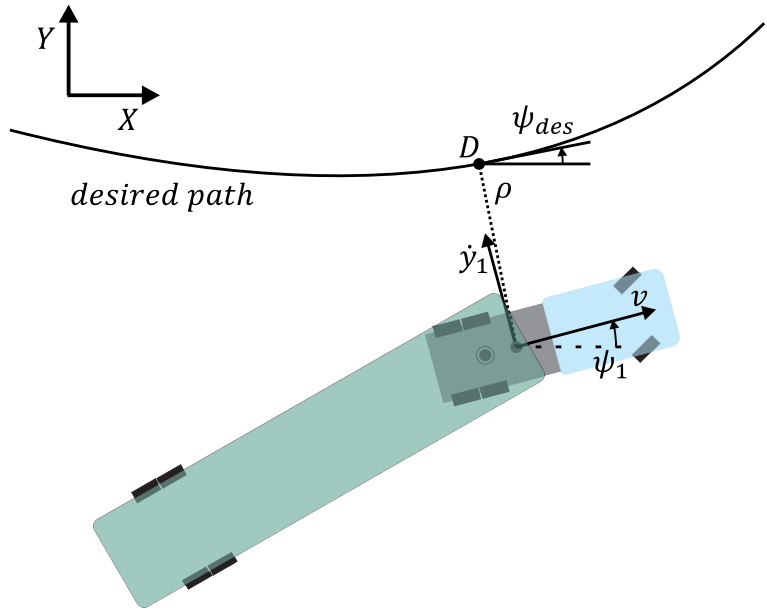}
      \caption{Schematic diagram for path-following model.}
      \label{path_following}
\end{figure}
%--

Based on Serret-Frenet equations \cite{skjetne2001nonlinear}, the path-following model of the autonomous ground vehicle is expressed as 
%--
\begin{equation}
\label{path-following-equation-1}
\begin{aligned}
\dot{\rho} &= v\sin\theta+\dot{y}_{1}\cos\theta\\
\dot{\theta} &= \dot{\psi}.
\end{aligned}
\end{equation}
%--
The displacement error $\rho$ can be rewritten in the linear form by assuming that the orientation error $\theta$ is small, as follows
%--
\begin{equation}
\label{path-following-equation-2}
\dot{\rho} = v\theta+\dot{y}_{1}.
\end{equation}
%--

\subsection{Articulated vehicle model}
\label{single-track-model}
Single-track models are widely used in literature \cite{kati2016robust,jujnovich2013path,Alcala2018,Ji2018} to describe the vehicle lateral behaviour without much modelling and parametrisation effort \cite{schramm2014vehicle}. These assume that the vehicle can be described by only one equivalent track in each axle, linked by the vehicle body. Consequently, it only takes into account the planar movement of the vehicle, disregarding roll and pitch effects. The nonholonomic linear model adopted here is based on bicycle model presented by van de Molengraft-Luijten \textit{et al.} \cite{van2012analysis}.

\autoref{single-track} shows the
free body diagram of a vehicle with one articulation, where the following assumptions are adopted:
%--
\begin{itemize}
\item Differences between left and right track are ignored;
\item Vehicle velocity parameter is constant;
\item The mass of each unit is assumed to be concentrated at the centre of gravity;
\item Lateral tyre forces are proportional to the tyre slip angles;
\item There is no load transfer.
\end{itemize}
%--
\begin{figure}[h]
      \centering
      \includegraphics[scale=0.3]{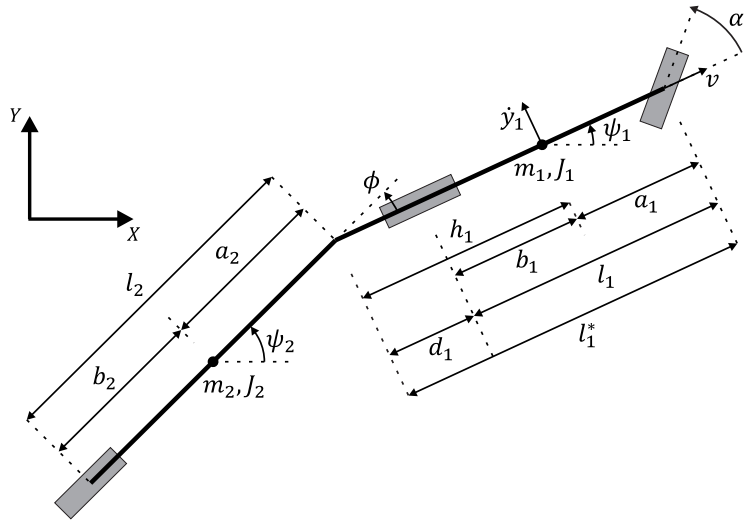}
      \caption{Articulated vehicle single-track model.}
      \label{single-track}
\end{figure}

\autoref{parameters} details the parameters of the articulated vehicle shown in \autoref{single-track}. Note that hitch point may be positioned behind the towing vehicle rear axle (e.g. truck-full trailers where $h_{1}>b_{1}$ and $d_{1}>0$) or in front of it (e.g. tractor-semitrailers where $h_{1}<b_{1}$ and $d_{1}<0$).
%--
\begin{table}[H]
\caption{Description of vehicle parameters}
\label{parameters}
\begin{center}
\begin{tabular}{c l c}
\hline
Parameter & Meaning & Unit\\
\hline
$a_{1}$ & Distance from the front axle to the tractor centre of gravity & $m$\\
$a_{2}$ & Distance from the coupling point to the trailer centre of gravity & $m$\\
$b_{1}$ & Distance from the tractor rear axle to the tractor centre of gravity & $m$\\
$b_{2}$ & Distance from the trailer axle to the trailer centre of gravity & $m$\\
$l_{1}$ & Tractor wheelbase & $m$\\
$l_{2}$ & Trailer wheelbase & $m$\\
$d_{1}$ & The distance between the tractor rear axle and the coupling point & $m$\\
$h_{1}$ & The distance between coupling point and the tractor centre of gravity & $m$\\
$l_{1}^{*}$ & The distance between the tractor front axle and the coupling point & $m$\\
$v$ & Forward velocity & $m/s$\\
$\dot{y}_{1}$ & Lateral Velocity & $m/s$\\
$m_{1}$ & Tractor mass & $kg$\\
$m_{2}$ & Trailer mass & $kg$\\
$J_{1}$ & Tractor moment of inertia & $kg \hspace{1mm} m^{2}$\\
$J_{2}$ & Trailer moment of inertia & $kg \hspace{1mm} m^{2}$\\
$\psi_{1}$ & Tractor yaw & $rad$\\
$\psi_{2}$ & Trailer yaw & $rad$\\
$\alpha$ & Steering angle & $rad$\\
$\phi$ & Articulation angle & $rad$\\
\hline
\end{tabular}
\end{center}
\end{table}
%--

In order to improve the path following, lateral displacement $\rho$ and orientations error $\theta$ must be as small as possible. In addition, it is necessary to ensure vehicle stability. Hence, the lateral velocity $\dot{y}_{1}$, yaw rate $\dot{\psi}_{1}$, articulation angle rate $\dot{\phi}$ and articulation angle $\phi$ must be well controlled.

The motion equation of the articulated vehicle can be expressed as
%--
\begin{equation}
M\dot{x} = Ax+B\alpha,
\end{equation}
%--
with the state vector defined as $x = [\dot{y}_{1}, \dot{\psi}_{1},\dot{\phi}, \phi, \rho, \theta]^{T}$. Therefore, the state-space description of the path-following model for the articulated heavy-duty vehicle is written as
%--
\begin{multline}
\label{state-space}
\begin{bmatrix}
m_{1}+m_{2} & -m_{2}(h_{1}+a_{2}) & -m_{2}a_{2} & 0 & 0 & 0\\
-m_{2}h_{1} & J_{1}+m_{2}h_{1}(h_{1}+a_{2}) & m_{2}h_{1}a_{2} & 0& 0 & 0\\
-m_{2}a_{2} & J_{2}+m_{2}a_{2}(h_{1}+a_{2}) & J_{2}+m_{2}a_{2}^{2} & 0& 0 & 0\\
0 & 0 & 0 & 1 & 0 & 0\\
0 & 0 & 0 & 0 & 1 & 0\\
0 & 0 & 0 & 0 & 0 & 1
\end{bmatrix}
\dot{x}
=\\
\begin{bmatrix}
\frac{-c_{1}-c_{2}-c_{3}}{v} & \frac{c_{3}(h_{1}+l_{2})-a_{1}c_{1}+b_{1}c_{2}-(m_{1}+m_{2})v^{2}}{v} & \frac{c_{3}l_{2}}{v} & c_{3} & 0 & 0\\
\frac{c_{3}h_{1}-a_{1}c_{1}+b_{1}c_{2}}{v} & \frac{m_{2}h_{1}v^{2}-a_{1}^{2}c_{1}-b_{1}^{2}c_{2}-c_{3}h_{1}(h_{1}+l_{2})}{v} & \frac{-c_{3}h_{1}l_{2}}{v} & -c_{3}h_{1} & 0 & 0\\
\frac{c_{3}l_{2}}{v} & \frac{m_{2}a_{2}v^{2}-c_{3}l_{2}(h_{1}+l_{2})}{v} & \frac{-c_{3}l_{2}^{2}}{v} & -c_{3}l_{2} & 0 & 0\\
0 & 0 & 1 & 0 & 0 & 0\\
1 & 0 & 0 & 0 & 0 & v\\
0 & 1 & 0 & 0 & 0 & 0
\end{bmatrix}
x
+
\begin{bmatrix}
c_{1}\\
a_{1}c_{1}\\
0\\
0\\
0\\
0
\end{bmatrix}
\alpha,
\end{multline}
%--
where the vehicle steering angle $\alpha$ is the control input, $c_{1}$, $c_{2}$ and $c_{3}$ are the cornering stiffness of the tractor front axle, tractor rear axle and trailer axle, respectively.

Many studies consider constant cornering stiffness. However, this hypothesis is not considered in this work since these coefficients may vary according to several vehicle parameters. In fact, Fancher demonstrated in \cite{fancher1989directional} (as cited in \cite{luijten2010lateral}) that the relation between the tyre cornering stiffness and the vertical load forces are approximately linear for truck tyres. The coefficient of proportionality is given by a normalised cornering stiffness $f_{j}$, and the cornering stiffness $c_{j}$ scales linearly with the vertical load force of the axle $F_{z_{j}}$. Therefore, the cornering stiffness parameters are calculated as:
%--
\begin{equation}
\label{cornering-stiffnes}
c_{j} = f_{j}F_{z_{j}} \text{ with } j = 1,\hdots,p,
\end{equation}
%--
where $p$ is the number of axles in the vehicle, $j = 1$ corresponds to the tractor front axle, $j = 2$ to the tractor rear axle and $j = 3$ to the trailer axle. 

The vertical force in each axle can be calculated as
%--
\begin{equation}
\label{vertical-force}
\begin{aligned}
F_{z_{1}} &= m_{1}g\frac{b_{1}}{l_{1}}-m_{2}g\frac{b_{2}d_{1}}{l_{2}l_{1}}\\
F_{z_{2}} &= m_{1}g\frac{a_{1}}{l_{1}}+m_{2}g\frac{b_{2}l^{*}_{1}}{l_{2}l_{1}}\\
F_{z_{3}} &= m_{2}g\frac{a_{2}}{l_{2}},
\end{aligned}
\end{equation}
%--
where $g$ is the gravitational acceleration. Moreover, Houben \cite{houben2008analysis} (as cited in \cite{van2012analysis}) observed that the normalised cornering stiffness of trailer tyres, drive and steer are approximately the same. Therefore, it is assumed $f_{1} \approx f_{2} \approx f_{3}$.

Nevertheless, a discrete state-space representation of the system is necessary in order to perform the robust recursive control for time-varying linear systems subject to parametric uncertainties. Hence, the System \eqref{state-space} is discretised by using the Tustin method. 

\section{Robust recursive regulator}\label{robust_recursive_regulator}

The goal of the Robust Linear Quadratic Regulator (RLQR) is to minimise a given cost function subject to the maximum influence of parametric uncertainties. It is made by implementing an optimal feedback law in the form $u_i=K_ix_i$, where $K_i$ is the feedback gain. This section describes the robust recursive regulator presented by Terra \textit{et al.} in \cite{Cerri2014} and Cerri \textit{et al.} in \cite{cerri2009recursive}. 

\subsection{Problem formulation}
Consider the following discrete-time linear system subject to parametric uncertainties
\begin{equation}
\label{eq:statespaceDT}
x_{i+1} = (F_{i} + \delta F_{i})x_{i} + (G_{i}  + \delta G_{i})u_{i},
\end{equation}
where $i = 0,\hdots, N$, $x_{i} \in \mathbb{R}^{n}$ is the state vector, $u_{i} \in \mathbb{R}^{m}$ is the control input, and $F_{i} \in \mathbb{R}^{n \times n}$ and $G_{i} \in \mathbb{R}^{n \times m}$ are known nominal model matrices. Uncertainty matrices $\delta F_{i}$ and $\delta G_{i}$ represent parametric uncertainties modelled as
\begin{equation}
\label{eq:RLQRuncertainties}
\begin{bmatrix}
\delta {F}_{i} & \delta G_{i}
\end{bmatrix} = H_{i} \Delta_{i} \begin{bmatrix}
E_{F_{i}} & E_{G_{i}}
\end{bmatrix},
\end{equation}
where $i = 0,\hdots, N$; $H_{i} \in \mathbb{R}^{n \times p}$; $E_{F_{i}} \in \mathbb{R} ^{l \times n}$ and $E_{G_{i}} \in \mathbb{R} ^{l \times m}$ are known matrices; and $\Delta_{i} \in \mathbb{R}^{p \times l}$ is an arbitrary matrix such that $||\Delta|| \leq 1$.

In order to obtain the Robust Linear Quadratic Regulator, the following optimisation problem must be solved \cite{Cerri2014}:
\begin{equation}
\label{eq:rlqrminmax}
\underset{x_{i+1},u_{i}}{min} \ \underset{{\delta} F_{i},{\delta} G_{i}}{max} {\bar{J}^{\mu}_{i}(x_{i+1},u_{i},{\delta} F_{i},{\delta} G_{i})},
\end{equation}
where $\bar{J}^{\mu}_{i}$ is the cost function
\begin{align}
\label{eq:rqlqrcostfun}
\nonumber
& \bar{J}^{\mu}_i(x_{i+1},u_{i},{\delta} F_{i},{\delta} G_{i}) = \\
& \begin{bmatrix}
x_{i+1}\\
u_i
\end{bmatrix}^T
\begin{bmatrix}
P^{r}_{i+1} & 0 \\
0 & R_i
\end{bmatrix} 
\begin{bmatrix}
x_{i+1}\\
u_i
\end{bmatrix} + \Phi^T 
\begin{bmatrix}
Q_i & 0\\
0 & \mu I
\end{bmatrix}\Phi,
\end{align}
with fixed penalty parameter $\mu > 0$, weighing matrices $Q_{i} \succ 0$, $R_{i} \succ 0$, $P_{i+1} \succ 0$ and
\begin{equation*}
\Phi = \left\{
\begin{bmatrix}
0 & 0 \\
I & -G_i-\delta G_i
\end{bmatrix}
\begin{bmatrix}
x_{i+1}\\
u_i
\end{bmatrix}
- \begin{bmatrix}
-I\\
F_i+{\delta} F_i
\end{bmatrix}
x_i\right\}.
\end{equation*}
Details on penalty function can be seen in \cite{cerri2009recursive}.

\begin{remark}
The optimisation problem (\ref{eq:rlqrminmax})-(\ref{eq:rqlqrcostfun}) is a particular case of the robust least-squares problem and will be treated below.
\end{remark}

\subsection{Regularised least squares}

Consider the least-square minimisation problem defined by
%--
\begin{equation}
\label{robust_regularized_min_sqr_prob_a}
\min_{x\in{\mathbb{R}^m}}\{J(x)\},
\end{equation}
%--
where $J(x)$ is a regularised quadratic functional
%--
\begin{equation}
\label{robust_regularized_min_sqr_prob_b}
\begin{split}
J(x)&=\|x\|^{2}_{Q}+\|Ax-b\|^{2}_{W}\\&=\,x^{T}Qx+(Ax-b)^TW(Ax-b),
\end{split}
\end{equation}
%--
with $Q \in \mathbb{R}^{m \times m}$ (regularisation matrix) and $W \in \mathbb{R}^{m \times n}$ symmetric positive definite, $A \in \mathbb{R}^{n \times n}$ and $b \in \mathbb{R}^{n}$ known, and $x \in \mathbb{R}^{m}$ the unknown vector.
%--
\begin{lemma}
\label{lemma1}
The optimal solution for the problem (\ref{robust_regularized_min_sqr_prob_a})-(\ref{robust_regularized_min_sqr_prob_b}) is
\begin{equation}
x^{\ast}=\left(Q+A^{T}WA\right)^{-1}A^{T}Wb. \nonumber
\end{equation}
\end{lemma}
%--
\begin{proof}
See \cite{kailath2000linear}.
\end{proof}

\subsection{Robust regularised least-squares problem}
In the regularised least-squares problem established in (\ref{robust_regularized_min_sqr_prob_a})-(\ref{robust_regularized_min_sqr_prob_b}), now suppose that the matrix $A$ and the vector $b$ are under influence of uncertainties $\delta A$ and $\delta b$, respectively. Consider the min-max optimisation problem defined in \cite{sayed1999design} in the form: 
%--
\begin{equation}
\label{rob_reg_least_sqr_prob}
\min_{x} \max_{\delta{A},\delta{b}}\{J(x,\delta{A},\delta{b})\},
\end{equation}
%--
with $J(x,\delta{A},\delta{b})$ given by
%--
\begin{equation}
\label{func_min_quad_reg_com_inct}
\begin{array}{c}
J(x,\delta{A},\delta{b})=\|x\|^{2}_{Q}+\|(A+\delta{A})x-(b+\delta{b})\|^{2}_{W},
\end{array}
\end{equation}
%--
and the uncertainties $\delta A$ and $\delta b$ modelled as
\begin{equation}
\label{uncertain_model}
\begin{bmatrix}
\,\delta{A} & \delta{b\,}
\end{bmatrix}
=H\Delta{\begin{bmatrix} E_{A} & E_{b}
\end{bmatrix}},
\end{equation}
with $A$, $b$, $H$, $E_{A}$, $E_{b}$, $Q$ and $W$ known matrices, $\Delta$ a contraction arbitrary matrix ($\|\Delta\|\leq{1}$) and $x$ an unknown vector. The optimal solution for the problem (\ref{rob_reg_least_sqr_prob})-(\ref{uncertain_model}) is given below. See demonstration details in \cite{sayed1999design}, where a general result is proposed.
%--
\begin{lemma}
\label{lemma_lambda}
The optimisation problem (\ref{rob_reg_least_sqr_prob})-(\ref{uncertain_model}) has a unique solution
%--
\begin{equation}
\label{sol_oti_RLS_min_max_unica}
x^{\ast}=\left(\hat{Q}+A^{T}\hat{W}A\right)^{-1}\left(A^{T}\hat{W}b+\hat{\lambda}E_{A}^{T}E_{b}\right), \nonumber
\end{equation}
%--
with $\hat{Q}$ and $\hat{W}$ defined as
%--
\begin{eqnarray}
\hat{Q}&:=&Q+\hat{\lambda}{E_{A}^{T}E_{A}},  \nonumber \\
\hat{W}&:=&W+WH(\hat{\lambda}I-H^{T}WH)^{\dagger}H^{T}W.  \nonumber
\end{eqnarray}
%--
The non-negative scalar parameter obtained from the minimisation problem 
%--
\begin{equation} \label{lambda_otimo}
\hat{\lambda}\, = \,\arg\min_{\lambda\geq{\|H^{T}WH\|}}\left\{\Gamma(\lambda)\right\},  \nonumber
\end{equation}
%--
where $\Gamma(\lambda)\,:=\,\|x(\lambda)\|^{2}_{Q}+\lambda\|E_{A}x(\lambda)-E_{b}\|^{2}+\|Ax(\lambda)-b\|^{2}_{W(\lambda)}$ with
%--
\begin{eqnarray}
{Q}(\lambda)\,&:=&\,Q+{\lambda}{E_{A}^{T}E_{A}},  \nonumber \\
{\tilde Q}(\lambda)\,&:=&{Q}(\lambda)+A^{T}{W}(\lambda)A,  \nonumber \\
{W}(\lambda)\,&:=&\,W+WH({\lambda}I-H^{T}WH)^{\dagger}H^{T}W,  \nonumber
\\
x(\lambda)\,&:=&\,\tilde Q(\lambda)^{-1}\left(A^{T}{W}(\lambda)b+{\lambda}E_{A}^{T}E_{b}\right). \nonumber
\end{eqnarray}
%--
\begin{proof}
See \cite{sayed1999design}
\end{proof}
\end{lemma}
%--

For this type of problem, it is appropriate to redefine \autoref{lemma_lambda} in terms of an array of matrices. The following lemma shows an optimal solution for the problem (\ref{rob_reg_least_sqr_prob})-(\ref{uncertain_model}) in an alternative structure to this fundamental theorem.
%--
\begin{lemma}
\label{lemma2}
Suppose $Q\succ{0}$ and $W\succ{0}$. The solution $x^{*}$ for the problem (\ref{rob_reg_least_sqr_prob})-(\ref{uncertain_model}) can be rewritten as 
%--
\begin{equation}
\setlength{\arraycolsep}{1pt}
\begin{bmatrix}
x^{\ast} \\
J(x^{\ast})
\end{bmatrix}=\begin{bmatrix}
0 & 0 \\ 0 & b \\ 0 & E_{b} \\ I & 0 \end{bmatrix}^{T}\begin{bmatrix}
Q^{-1}&0&0&I\\
0&\hat{W}^{-1}&0&A\\
0&0&\hat{\lambda}^{-1}I&E_A\\
I&A^T&E_A^T&0
\end{bmatrix}^{-1}
\begin{bmatrix}
0\\
b\\
E_b\\
0
\end{bmatrix}, \nonumber
\end{equation}
with $\hat{W}$ and $\hat{\lambda}$ as in \autoref{lemma_lambda}.
%--
\end{lemma}
%--
\begin{proof}
See \cite{cerri2009recursive}.
\end{proof}
\subsection{Robust Linear Quadratic Regulator}

The optimisation problem (\ref{eq:rlqrminmax})-(\ref{eq:rqlqrcostfun}) is solved based on the solution of a general robust regularised least-squares problem \cite{Cerri2014}. Back to the solution presented in \autoref{lemma_lambda}, with $\mu > 0$, the RLQR has an optimal operation point for each step $k$ of the algorithm. When suitable identifications of (\ref{eq:rlqrminmax})-(\ref{eq:rqlqrcostfun}) with (\ref{rob_reg_least_sqr_prob})-(\ref{uncertain_model}) are carried out, the regularisation of the robust regulator is reached thanks to minimisation over both $x_{i+1}(\mu)$ and $u_{i}(\mu)$ \cite{Cerri2014}:
%--
\begin{equation}
\label{identifications}
\setlength{\arraycolsep}{1pt}
Q\leftarrow{\begin{bmatrix}
P_{i+1} & 0\\
0 & R_{i}
\end{bmatrix}},\,\,\,x\leftarrow{\begin{bmatrix}
x_{i+1}(\mu)\\
u_{k}(\mu)
\end{bmatrix}}, 
W\leftarrow{\begin{bmatrix}
Q_{i} & 0\\
0 & \mu{I}
\end{bmatrix}},\,\,\,\nonumber
\end{equation}
\begin{equation}
\setlength{\arraycolsep}{1pt}
A\leftarrow{\begin{bmatrix}
0 & 0\\
I & -G_{i}
\end{bmatrix}},\,\,\,\delta{A}\leftarrow{\begin{bmatrix}
0 & 0\\
0 & -\delta G_{i}
\end{bmatrix}},\,\,\, \Delta \leftarrow \Delta_{i},\,\,\,\nonumber
\end{equation}\begin{equation}\setlength{\arraycolsep}{1pt}b\leftarrow{\begin{bmatrix}
-I\\
F_{i}
\end{bmatrix}
x_{i}},\,\,\,\delta{b}\leftarrow{\begin{bmatrix}
0\\
\delta F_{i}
\end{bmatrix}
x_{i}}, \nonumber
\end{equation}
\begin{equation} \setlength{\arraycolsep}{1pt}\label{rlqr_ident}
H\leftarrow{\begin{bmatrix}
0\\
H_{i}
\end{bmatrix}},\,\,\,
E_{A}\leftarrow{\begin{bmatrix} 0 & -E_{G_{i}}\end{bmatrix}},\,\,\,E_{b}\leftarrow{E_{F_{i}}x_{i}},
\end{equation}
%--

The following theorem shows a framework given in terms of an array of matrices with the purpose of calculating the optimal cost function, control input and state trajectory.
%--
\begin{theorem}
\label{theorem2}
For each $\mu > 0$ in the optimisation problem (\ref{eq:rlqrminmax})-(\ref{eq:rqlqrcostfun}), the optimal solution is given by
%--
\begin{equation} 
\label{exp_a}
\begin{bmatrix}
x^{\ast}_{i+1}(\mu)\\
u^{\ast}_{i}(\mu)\\
\tilde{J}^{\mu}_{i}(x^{\ast}_{i+1}(\mu),u^{\ast}_{i}(\mu))
\end{bmatrix}=\begin{bmatrix}
I & 0 & 0 \\
0 & I & 0 \\
0 & 0 & x_{i}(\mu)^{T}
\end{bmatrix}^{T}\begin{bmatrix}
L_{i,\mu}\\
K_{i,\mu}\\
P_{i,\mu}
\end{bmatrix}x_{i},
\end{equation}
%--
where the closed-loop system matrix $L_i$ and the feedback gain $K_i$ result from the recursion
\begin{equation}
\label{eq:rlqrsolution}
\begin{bmatrix}
L_{i} \\ K_{i} \\ P_{i}
\end{bmatrix} =
\begin{bmatrix}
0 & 0 & -I & \mathcal{F}_{i} & 0 & 0\\
0 & 0 & 0 & 0 & 0 & I\\ 
0 & 0 & 0 & 0 & I & 0
\end{bmatrix}
\Xi^{-1}
\begin{bmatrix}
0 \\ 0 \\ -I \\ \mathcal{F}_{i} \\ 0 \\ 0
\end{bmatrix},
\end{equation}
with
\begin{equation}
\Xi = \begin{bmatrix}
P^{-1}_{i+1} & 0 & 0 & 0 & I & 0 \\
0 & R^{-1}_{i} & 0 & 0 & 0 & I \\
0 & 0 & Q^{-1}_{i} & 0 & 0 & 0 \\
0 & 0 & 0 & \Sigma_{i}\left(\mu,\hat \lambda_{i}\right) & \mathcal{I} & -\mathcal{G}_{i} \\
I & 0 & 0 & \mathcal{I}^{T} & 0 & 0 \\
0 & I & 0 & -\mathcal{G}^{T} & 0 & 0
\end{bmatrix}, \nonumber
\end{equation}
\begin{align*}
& \Sigma_{i} = \begin{bmatrix}
\mu^{-1}I - \hat{\lambda}^{-1}_{i}H_{i}H_{i}^{T}  & 0 \\
0 & \hat{\lambda}^{-1}_{i}I
\end{bmatrix}, \nonumber \\
& \mathcal{I} = \begin{bmatrix}
I \\ 0
\end{bmatrix}, \
\mathcal{G}_{i} = \begin{bmatrix}
G_{i} \\ E_{G_{i}}
\end{bmatrix}, \
\mathcal{F}_{i} = \begin{bmatrix}
F_{i} \\ E_{F_{i}}
\end{bmatrix},
\end{align*}
%--
\noindent where $P_{i+1}$ is the solution of the associated Riccati Equation and $\lambda_i> \|\mu H_i^TH_i\|$ \cite{Sayed2001}. Furthermore, alternatively one has
%--

\begin{equation}
\label{alternative}
\begin{aligned}
P_{i,\mu} = &L^{T}_{i,\mu}P_{i+1}L_{i,\mu} + K_{i,\mu}R_{i}K_{i,\mu} + Q_{i}+\\
&(\mathcal{I}L_{i,\mu} - \mathcal{G}_{i}K_{i,\mu} - \mathcal{F}_{i})^{T}\Sigma_{i,\mu}^{-1}(\mathcal{I}L_{i,\mu}-\mathcal{G}_{i}K_{i,\mu}-\mathcal{F}_{i}) \succ 0.
\end{aligned}
\end{equation}
%--
\end{theorem}
%--
\begin{proof}
It follows from \autoref{lemma2}, identifications performed in (\ref{identifications}) and results shown in \cite{cerri2009recursive}.
\end{proof}
%--
\autoref{rlqr_algor} shows the Robust Linear Quadratic Regulator obtained with \autoref{lemma_lambda}. The parameter $\mu$ is associated with system robustness. It is responsible for ensuring the RLQR regularisation and validity of the equality (\ref{eq:statespaceDT}). For maximum robustness, $\mu \rightarrow \infty$ and consequently $\Sigma_i \rightarrow 0$.
%---------------------------------------------------------------------------------------------------------------------------------------
\vspace{2mm}
\begin{center}
\begin{algorithm}[H]
\label{rlqr_algor}
\SetAlgoLined
\textbf{Uncertain model:} Consider the model (\ref{eq:statespaceDT})-(\ref{eq:RLQRuncertainties}) and criterion (\ref{eq:rlqrminmax})-(\ref{eq:rqlqrcostfun}) with known\\ $F_{i}$, $G_{i}$, $E_{F_{i}}$, $E_{G_{i}}$, $Q_{i} \succ 0$, and $R_{i} \succ 0$ for all $i$.\\
\textbf{Initial conditions:} Define $x_{0}$ and $P_{i,N}\succeq{0}$.\\
\textbf{Step 1:} \textit{(Backward)} For all $i=N-1,\ldots,0$, compute\\
$\hfill\begin{bmatrix}
L_{i}\\
K_{i}\\
P_{i}
\end{bmatrix}
=
\begin{bmatrix}
0 & 0 & 0 \\
0 & 0 & 0 \\
0 & 0 & -I \\
0 & 0 & F_{i} \\
0 & 0 & E_{F_{i}} \\
I & 0 & 0 \\
0 & I & 0
\end{bmatrix}^{T}
\begin{bmatrix}
P_{i+1}^{-1}& 0 & 0 & 0 & 0 & I & 0\\
0 & R_{i}^{-1} & 0 & 0 & 0 & 0 & I\\
0 & 0 & Q_{i}^{-1} & 0 & 0 & 0 & 0\\
0 & 0 & 0 & 0 & 0 & I & -G_{i}\\
0 & 0 & 0 & 0 & 0 & 0 & -E_{G_{i}}\\
I & 0 & 0 & I & 0 & 0 & 0\\
0 & I & 0 & -G_{i}^{T} & -E_{G_{i}}^{T} & 0 & 0
\end{bmatrix}^{-1}
\begin{bmatrix}
0\\
0\\
-I\\
F_{i}\\
E_{F_{i}}\\
0\\
0
\end{bmatrix}.\hfill$\\
\textbf{Step 2:} \textit{(Forward)} For each $i=0,...,N-1$, obtain\\
$\hfill\begin{bmatrix}
x^{*}_{i+1}\\
u^{*}_{i}
\end{bmatrix}
=
\begin{bmatrix}
L_{i}\\
K_{i}
\end{bmatrix}
x^{*}_{i},\hfill$ \\
with the total cost given by $J_{r}^{*}=x_{0}^{T}P_{0}x_{0}$.
\\
 \caption{The Robust Linear Quadratic Regulator}
\end{algorithm}
\end{center}
\vspace{2mm}
%---------------------------------------------------------------------------------------------------------------------------------------
%--
For each iteration of (\ref{alternative}), the matrix $P_{i,\mu}$ is finite and $\mathcal{I}L_{i,\mu}-\mathcal{G}_{i}K_{i,\mu}-\mathcal{F}_{i} \rightarrow 0$, as shown in \cite{Cerri2014}. Therefore,
%--
\begin{equation}
\label{feedback}
\begin{aligned}
&L_{i,\infty} = F_{i}+G_{i}K_{i,\infty}\\
&E_{F_{i}} + E_{G_{i}}K_{i,\infty} = 0,
\end{aligned}
\end{equation}
%--
and a sufficient condition that satisfy (\ref{feedback}) is
%--
\begin{equation}
\label{eq:rank_cond}
rank\,\big(\begin{bmatrix}
E_{F_{i}} & E_{G_{i}}
\end{bmatrix}\big) = rank\,\big(E_{G_{i}}\big).
\end{equation}
%--
Convergence and stability analyses of the RLQR are made through direct identifications with the standard optimal regulator problem for systems not subject to uncertainties. It resembles the standard LQR where the stability is directly related with the positiveness of the Riccati equation solution \cite{Cerri2014}.
More details on convergence and stability analysis can be found in~\cite{Cerri2014}.

\section{Numerical results and discussion }\label{application_and_results}

For the controller validation, the RLQR was performed and compared with the $\mathcal{H}_{\infty}$ controller in various operational conditions. The Matlab/Simulink simulation software was used for this purpose. Simulations consist of minimising the lateral displacement and orientation errors. A double lane-change manoeuvre was performed during $30$ seconds with the sampling period being 0.01 seconds and the nominal payload subject to uncertainties. \autoref{lane-change} shows the scenario of studied cases, where $\varepsilon$ is the tractor width. Furthermore, \autoref{parameters-values} shows the vehicle parameters and the necessary information to calculate it, obtained from websites for commercial vehicles\footnote{https://www.scania.com} and towing implement\footnote{http://www.librelato.com.br} manufacturers. For all cases, the initial conditions are the same for both controllers, those being $x_{0} = [0,0,0,0,0.3,-0.1]^{T}$, the penalty parameter $\mu=10^{8}$, 
%--
\begin{equation*}
H =
\begin{bmatrix}
1\\
1\\
1\\
1\\
1\\
1
\end{bmatrix},\hspace{1mm}E_{F} = 
\begin{bmatrix}
6.8572\times 10^{-5}\\
-8.6201\times 10^{-5}\\
-2.1440\times 10^{-5}\\
-10.4924\times 10^{-5}\\
0\\
-666.66667\times 10^{-5}
\end{bmatrix}^{T}\textnormal{and} \hspace{1mm}E_{G} =
\begin{bmatrix}
-666.66667\times 10^{-5}\\
-666.66667\times 10^{-5}
\end{bmatrix}^{T}.
\end{equation*}
%--
The methodology used to calculate the uncertainty matrices is described in \ref{sec:uncertainties}.

The uncertainty parameters were chosen as vectors, this implies that the condition of existence of the controller (\ref{eq:rank_cond}) is satisfied, regardless of the numerical values of the uncertainty parameters.
%--
\begin{figure}[H]
      \centering
      \includegraphics[scale=0.55]{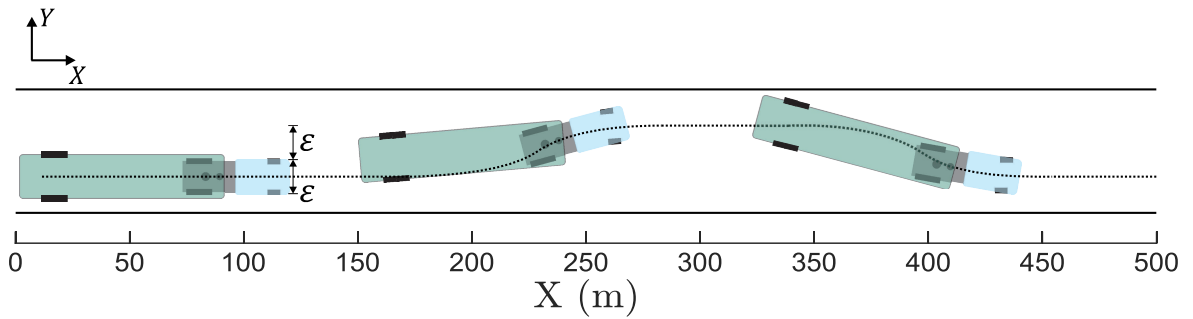}
      \caption{Double lane change scenario.}
      \label{lane-change}
\end{figure}
\begin{table}[H]
\caption{Vehicle parameters values}
\label{parameters-values}
\begin{center}
\begin{tabular}{c l}
\hline
Parameter & Value\\
\hline
$a_{1}$  & $1.734\hspace{1mm}m$\\
$a_{2}$  & $4.8\hspace{1mm}m$\\
$b_{1}$  & $2.415\hspace{1mm}m$\\
$b_{2}$ & $3.2\hspace{1mm}m$\\
$l_{1}$ & $4.149\hspace{1mm}m$\\
$l_{2}$ & $8.0\hspace{1mm}m$\\
$d_{1}$ & $-0.29\hspace{1mm}m$\\
$h_{1}$ & $2.125\hspace{1mm}m$\\
$l_{1}^{*}$ & $3.859\hspace{1mm}m$\\
$\varepsilon$ & $2.6 \hspace{1mm} m$\\
$v$ & $16.667\hspace{1mm}m/s$\\
$m_{1}$ & $8909\hspace{1mm}kg$\\
$m_{2}$ & $9370\hspace{1mm}kg$\\
Payload & $24000\hspace{1mm}kg$\\
$J_{1}$ & $41566\hspace{1mm}kg \hspace{1mm} m^{2}$\\
$J_{2}$ & $404360\hspace{1mm}kg \hspace{1mm} m^{2}$\\
$c_{1}$ & $345155\hspace{1mm}N/rad$ \\
$c_{2}$ & $927126\hspace{1mm}N/rad$\\
$c_{3}$ & $1158008\hspace{1mm}N/rad$\\
\hline
\end{tabular}
\end{center}
\end{table}
% %--

The normalised cornering stiffness was applied in all cases studied here as it is a satisfactory representation for most applications and conditions \cite{luijten2010lateral}. Thus, it was calculated as a function of vertical load by assuming $f = f_{1} = f_{2} = f_{3} = 5.73\hspace{1mm}rad^{-1}$. In addition, the maximum steering angle ($0.44\hspace{1mm}rad$) was taken into account in numerical results.

The linear system must be rewritten in order to compare the $\mathcal{H}_{\infty}$ control and the robust recursive regulator presented in this paper. Thus, the robust control design considering the $\mathcal{H}_{\infty}$ method discussed by Hassib \textit{et al.} \cite{hassibi1999indefinite} was used. Its equations, identifications and formulation are given in \ref{h_infty_app}. 

\autoref{blocks-diagrams} gives the block diagrams for both control techniques, where $e_{i}$ is the error between the reference and output, and $x_{ref}$ is the reference state vector. Since we aim to minimise the state variable errors, the control law is $u_{i} = K_{i}e_{i}$. Both $e_{i}$ and $x_{ref}$ are obtained when a reference control signal is applied to the lateral model of the vehicle.
%--
\begin{figure}[H]
      \centering
      \begin{subfigure}[b]{0.5\textwidth}
      \centering
      \label{rlqr-diagram}
      \includegraphics[scale=0.3]{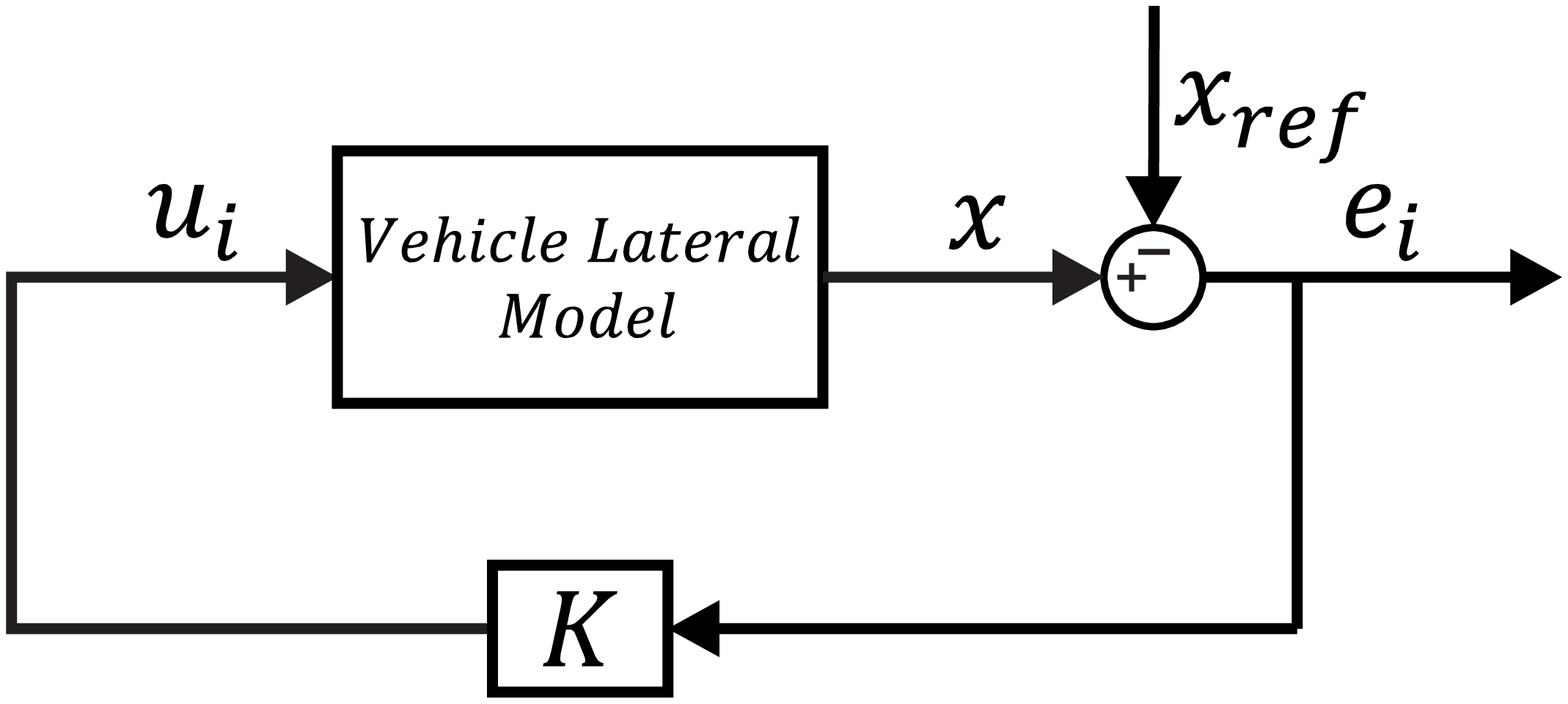}
      \caption{RLQR block diagram.}
      \end{subfigure}
      \begin{subfigure}[b]{0.45\textwidth}
      \centering
      \label{h_infinity-diagram}
      \includegraphics[scale=0.3]{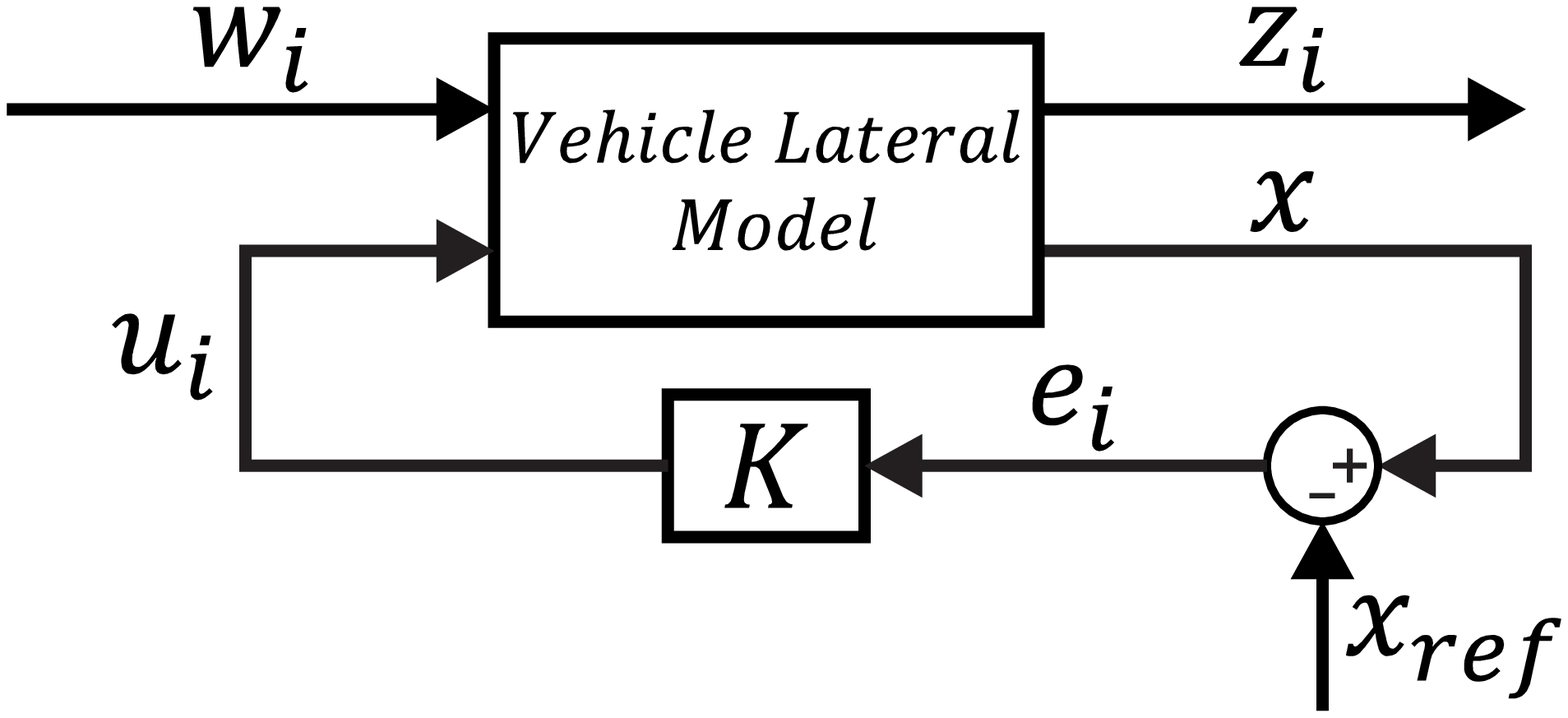}
      \caption{$\mathcal{H}_{\infty}$ control system block diagram.}
      \end{subfigure}
      \caption{Block diagrams for Robust Linear Quadratic Regulator and $\mathcal{H}_{\infty}$ control systems.}
      \label{blocks-diagrams}
\end{figure}
%--
\subsection{System response}

The articulated heavy vehicle behaviour was evaluated with numerical results by taking a given reference path. For this purpose, the lateral velocity, yaw rate, articulation angle rate, articulation angle, lateral displacement and orientation error of the vehicle were observed. Moreover, controller evaluation was done through graphic analysis, and by adopting maximum steering rate and $\mathcal{L}_2$ norm of the error as performance criteria.

\autoref{cases} and \autoref{l_2-norm} show maximum steering rate, payload variations, and $\mathcal{L}_2$ norm of lateral displacement and orientation errors for both performed controllers, respectively. Payload values for every evaluated case were chosen for the best illustration of the influence of mass variation.
%--
\begin{table}[H]
\caption{Evaluated cases for both controllers}
\label{cases}
\begin{center}
\begin{tabular}{c c c c}
\hline
Case  & Payload & $max\|\dot{\alpha}_{RLQR}\|$ & $max\|\dot{\alpha}_{\mathcal{H}_{\infty}}\|$ \\
\hline
1 & 100\% & 0.3432 rad/s & 4.3750 rad/s\\
2 & 234\% & 0.4130 rad/s& 8.4404 rad/s\\
3 & 237\% & 0.4164 rad/s & 9.2350 rad/s\\
4 & 0\% & 0.3333 rad/s & 4.5959 rad/s\\
\hline
\end{tabular}
\end{center}
\end{table}
%--
\begin{table}[H]
\caption{$\mathcal{L}_{2}$ norm of the lateral displacement and orientation errors}
\label{l_2-norm}
\begin{center}
\begin{tabular}{c c c c c}
\hline
 & \multicolumn{2}{c}{$\|\rho\|_{\mathcal{L}_{2}}$} & \multicolumn{2}{c}{$\|\theta\|_{\mathcal{L}_{2}}$}\\
\hline
Case & RLQR & $\mathcal{H}_{\infty}$ & RLQR & $\mathcal{H}_{\infty}$\\
 \hline
1 & 0.3727 & 0.2004 & 0.1481 & 0.0692\\
2 & 0.3886 & 0.1651 & 0.1331 & 0.0793\\
3 & 0.3882 & 0.4055 & 0.1328 & 0.2594\\
4 & 0.3217 & 0.2348 & 0.1358 & 0.0778\\
\hline
\end{tabular}
\end{center}
\end{table}
%--

The nominal payload was applied, and the weight matrices $Q$ and $R$ were adjusted so that the maximum steering rate was $max\|\dot{u}\|_{RLQR} \approx 0.3432\hspace{1mm} rad/s$. In addition, the direct counterpart weight matrices $R^{c}$ and $Q^{c}$ have the same values adjusted in $Q$ and $R$, respectively. Moreover, the robustness parameter $\gamma$ was adjusted to the lowest possible value that ensures $\mathcal{H}_{\infty}$ controller existence. Hence, for every case: 
\begin{gather*}
\gamma = 14350, \hspace{2mm} Q = R^{c} = 
\begin{bmatrix}
1 & 0 & 0 & 0 & 0 & 0\\
0 & 1 & 0 & 0 & 0 & 0\\
0 & 0 & 1 & 0 & 0 & 0\\
0 & 0 & 0 & 1 & 0 & 0\\
0 & 0 & 0 & 0 & 25000 & 0\\
0 & 0 & 0 & 0 & 0 & 100
\end{bmatrix} \text{ and }
R = Q^{c} = 
\begin{bmatrix}
67070 & 0\\
0 & 67070
\end{bmatrix}.
\end{gather*}
Graphics of numerical results for each evaluated case and their case descriptions follows.
\begin{description}
\item[Case 1]Considering nominal payload, \autoref{state-variables1} and \autoref{position-control1} show the system state variables, the global position of tractor centre of mass and steering angle performed by both controllers. 
%--
\begin{figure}[H]
  \centering
  \includegraphics[scale=0.34]{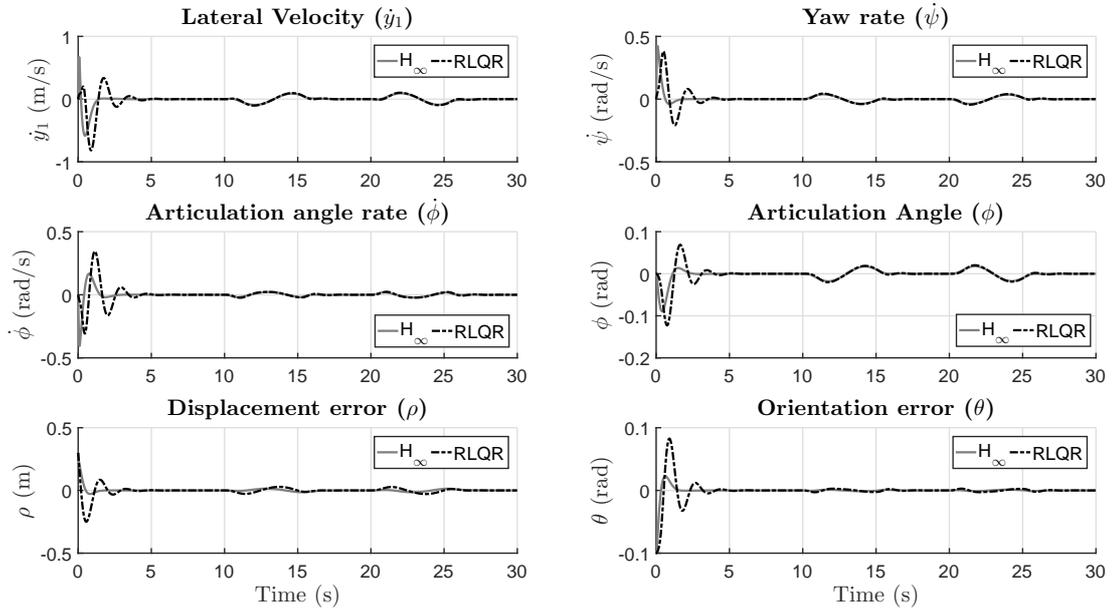}
  \caption{System state variables for case 1.}
  \label{state-variables1}
\end{figure}
%--
\begin{figure}[H]
  \centering
  \includegraphics[scale=0.35]{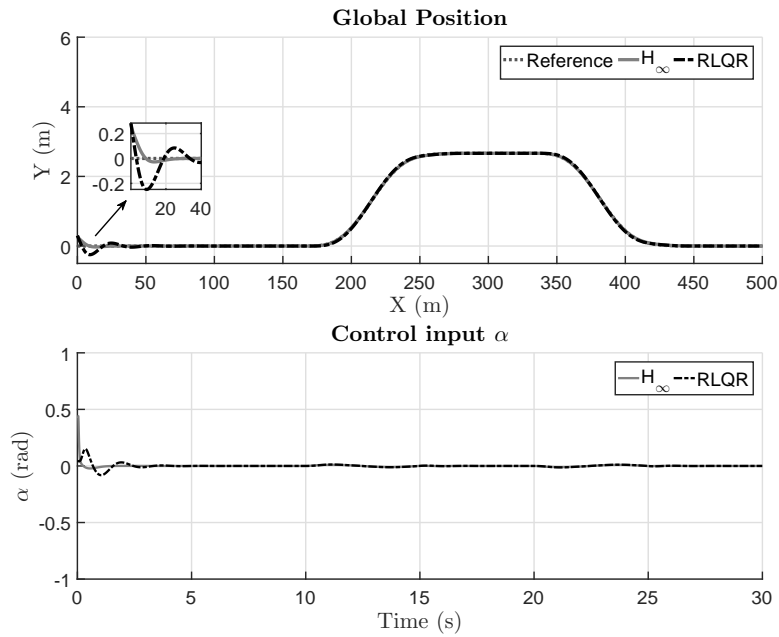}
  \caption{Global position of the tractor centre of mass and steering angle for case 1.}
  \label{position-control1}
\end{figure}
%--
\item[Case 2] Considering $234\%$ of overload over the payload nominal value, \autoref{state-variables2} and \autoref{position-control2} show the system state variables, the global position of tractor centre of mass and steering angle performed by both controllers.
%--
\begin{figure}[H]
  \centering
  \includegraphics[scale=0.34]{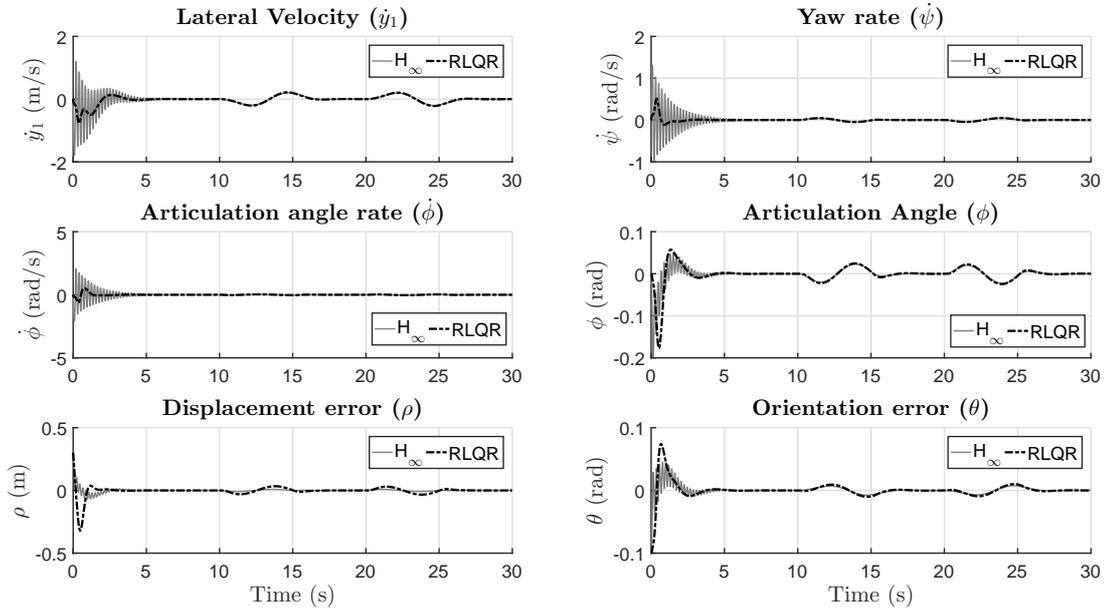}
  \caption{System state variables for case 2.}
  \label{state-variables2}
\end{figure}
%--
\begin{figure}[H]
  \centering
  \includegraphics[scale=0.35]{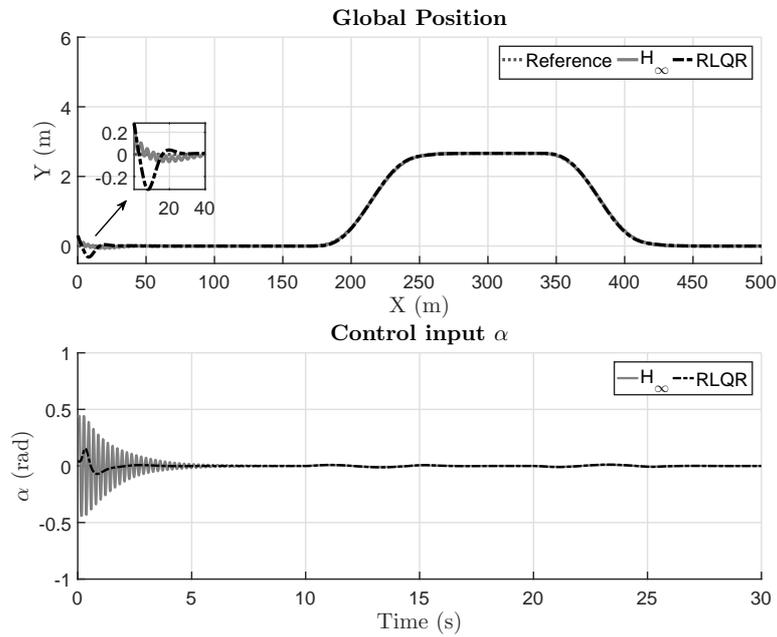}
  \caption{Global position of the tractor centre of mass and steering angle for case 2.}
  \label{position-control2}
\end{figure}
%--
\item[Case 3] Considering $237\%$ of overload over the payload nominal value, \autoref{state-variables3} and \autoref{position-control3} show the system state variables, the global position of tractor centre of mass and steering angle performed by both controllers.
\begin{figure}[H]
  \centering
  \includegraphics[scale=0.34]{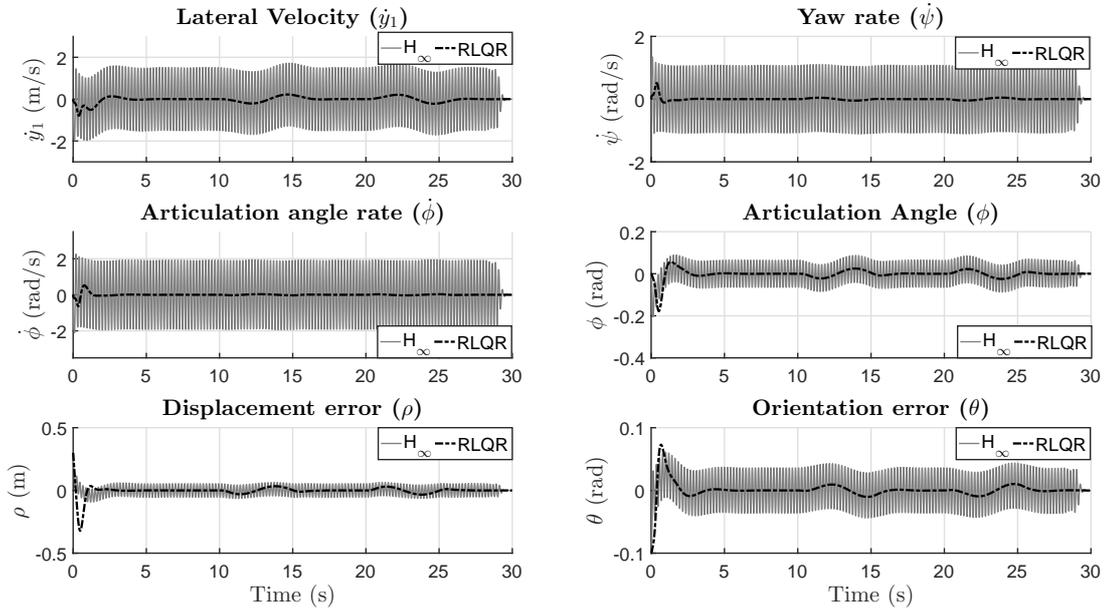}
  \caption{System state variables for case 3.}
  \label{state-variables3}
\end{figure}
%--
\begin{figure}[H]
  \centering
  \includegraphics[scale=0.6]{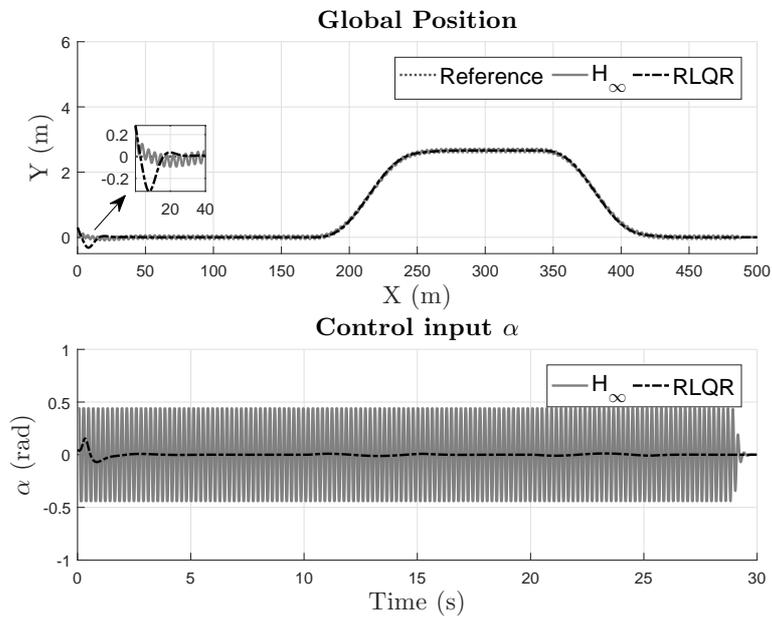}
  \caption{Global position of the tractor centre of mass and steering angle for case 3.}
  \label{position-control3}
\end{figure}
%--
\item[Case 4] Lastly, considering a vehicle without payload, \autoref{state-variables4} and \autoref{position-control4} show the system state variables, the global position of the tractor centre of mass and steering angle performed by both controller. 
\begin{figure}[H]
  \centering
  \includegraphics[scale=0.34]{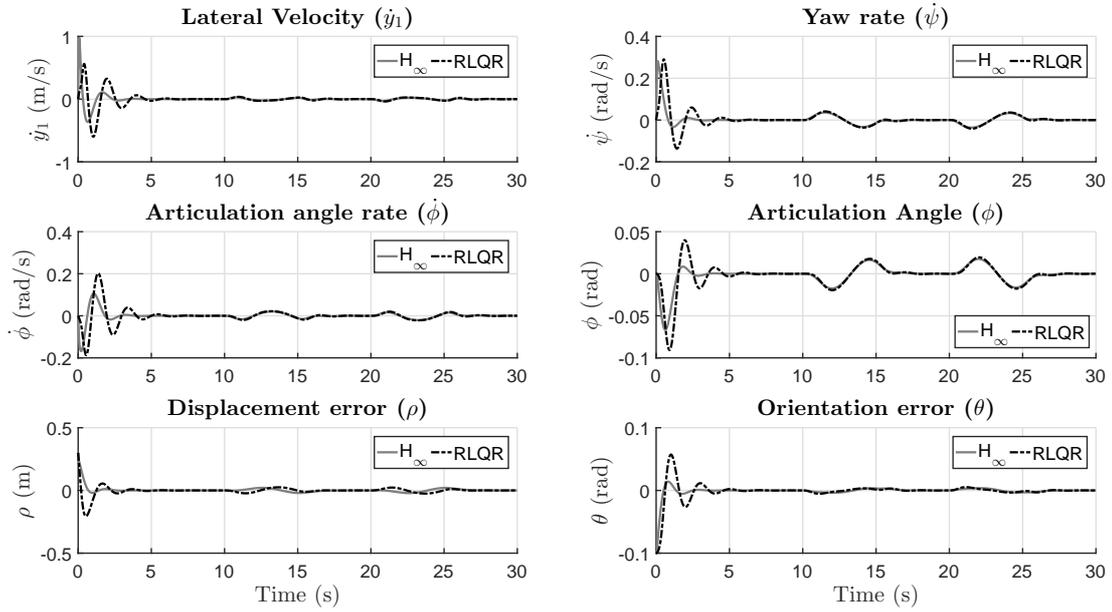}
  \caption{System state variables for case 4.}
  \label{state-variables4}
\end{figure}
%--
\begin{figure}[H]
  \centering
  \includegraphics[scale=0.35]{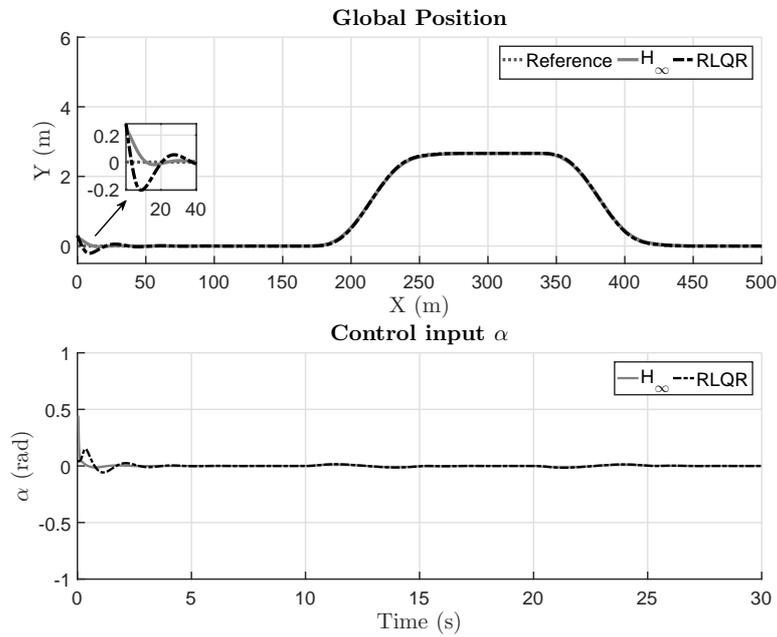}
  \caption{Global position of the tractor centre of mass and steering angle for case 4.}
  \label{position-control4}
\end{figure}
%--
\end{description}
\subsection{Discussion} 

The main goal of these evaluated cases was to show how the Robust Recursive Regulator deals with uncertainties in articulated heavy vehicles. Results demonstrate that the RLQR performance is less affected by payload mass variation than $\mathcal{H}_{\infty}$ controller. It is verified in \autoref{l_2-norm}, where the $\mathcal{L}_{2}$ norm of the lateral displacement and orientation errors of the robust recursive regulator are less affected by mass uncertainties than $\mathcal{H}_{\infty}$ controller.

Moreover, \autoref{cases} shows that, in the presence of uncertainties, the $max\|\dot{u}_{\mathcal{H}_{\infty}}\|$ is severely influenced by parametric variations while $max\|\dot{u}_{RLQR}\|$ is much less affected. This is significant, since high steering angle rates mean abrupt driving, which may not be possible for the mechanical system of the vehicle, representing a safety limitation to the $\mathcal{H}_{\infty}$ controller.

As shown in \autoref{state-variables1}, \autoref{state-variables2}, \autoref{state-variables3} and \autoref{state-variables4}, the performed results for tractor lateral velocity $\dot{y}_1$, tractor yaw rate $\dot{\psi}$, articulation angle rate $\dot{\phi}$ and articulation angle $\phi$ demonstrate that the robust recursive regulator deals better with vehicle lateral dynamic behaviour since its performance is much less affected by mass variations. Furthermore, the results obtained in \autoref{cases} and \autoref{l_2-norm} are shown in  \autoref{state-variables1}-\ref{position-control4}, confirming that the RLQR is still more robust, more stable, smoother and safer than $\mathcal{H}_{\infty}$ in the presence of payload variations. For better performance of the $\mathcal{H}_{\infty}$ controller, the parameter $\gamma$ needs to be adjusted offline for each particular payload. This is very inefficient for practical applications given the wide mass variation in heavy-duty vehicles. The advantage of the RLQR is that it does not require offline adjustment of auxiliary parameters since the penalty parameter $\mu$ ensures smoothness and robustness, maintaining the optimality and good performance for each evaluated case. 

\section{Conclusions}\label{conclusions}

The Robust Linear Quadratic Regulator has been applied to perform the lateral control of an autonomous articulated heavy-duty vehicle subject to parametric uncertainties. Considering uncertainty in the towed mass, RLQR controller performance was better in terms of robustness, lateral stability, driving smoothness and safety when compared to the $\mathcal{H}_{\infty}$ robust control technique. Thus, the robust recursive regulator was demonstrated as a profitable control technique to deal with parametric uncertainties in such vehicle systems. The RLQ controller performs well for a wide range of payloads, while the performance of the $H_{\infty}$ controller is significantly affected by higher payloads, given a constant $\gamma$. Nevertheless, vertical and roll stability cannot be guaranteed because a model-based control design that only considers planar motion was used. 

The robust recursive regulator could be exploited in non-articulated and multi-articulated vehicles in order to perform the path-following and lateral control. For future work, the articulated vehicle system will be extended to three-dimensions representation and experimental results will be obtained.

\section{Acknowledgements}
The authors would like to thank the Coordenação de Aperfeiçoamento de Pessoal de Nível Superior - Brasil - CAPES (Finance Code 001), São Paulo Research Foundation (FAPESP, grant \#2014/50851-0) and Vale S.A. for the financial support.
%\section*{References}
\bibliography{mybibfile}

\begin{thebibliography}{10}
\expandafter\ifx\csname url\endcsname\relax
  \def\url#1{\texttt{#1}}\fi
\expandafter\ifx\csname urlprefix\endcsname\relax\def\urlprefix{URL }\fi
\expandafter\ifx\csname href\endcsname\relax
  \def\href#1#2{#2} \def\path#1{#1}\fi

\bibitem{wu_2015}
N.~Wu, W.~Huang, Z.~Song, X.~Wu, Q.~Zhang, S.~Yao,
  \href{https://doi.org/10.1109/ivs.2015.7225817}{Adaptive dynamic preview
  control for autonomous vehicle trajectory following with {DDP} based path
  planner}, in: 2015 {IEEE} Intelligent Vehicles Symposium ({IV}), {IEEE},
  2015, pp. 1012--1017.
\newblock \href {http://dx.doi.org/10.1109/ivs.2015.7225817}
  {\path{doi:10.1109/ivs.2015.7225817}}.
\newline\urlprefix\url{https://doi.org/10.1109/ivs.2015.7225817}

\bibitem{fu_2015}
M.~Fu, K.~Zhang, Y.~Yang, H.~Zhu, M.~Wang,
  \href{https://doi.org/10.1109/ivs.2015.7225800}{Collision-free and
  kinematically feasible path planning along a reference path for autonomous
  vehicle}, in: 2015 {IEEE} Intelligent Vehicles Symposium ({IV}), {IEEE},
  2015, pp. 907--912.
\newblock \href {http://dx.doi.org/10.1109/ivs.2015.7225800}
  {\path{doi:10.1109/ivs.2015.7225800}}.
\newline\urlprefix\url{https://doi.org/10.1109/ivs.2015.7225800}

\bibitem{Shin2015}
J.~Shin, J.~Huh, Y.~Park,
  \href{https://doi.org/10.1016/j.conengprac.2015.03.006}{Asymptotically stable
  path following for lateral motion of an unmanned ground vehicle}, Control
  Engineering Practice 40 (2015) 102--112.
\newblock \href {http://dx.doi.org/10.1016/j.conengprac.2015.03.006}
  {\path{doi:10.1016/j.conengprac.2015.03.006}}.
\newline\urlprefix\url{https://doi.org/10.1016/j.conengprac.2015.03.006}

\bibitem{masserafilho_2017}
C.~M. Filho, M.~H. Terra, D.~F. Wolf,
  \href{https://doi.org/10.1109/tits.2017.2679098}{Safe optimization of highway
  traffic with robust model predictive control-based cooperative adaptive
  cruise control}, {IEEE} Transactions on Intelligent Transportation Systems
  18~(11) (2017) 3193--3203.
\newblock \href {http://dx.doi.org/10.1109/tits.2017.2679098}
  {\path{doi:10.1109/tits.2017.2679098}}.
\newline\urlprefix\url{https://doi.org/10.1109/tits.2017.2679098}

\bibitem{dias_2015}
J.~E.~A. Dias, G.~A.~S. Pereira, R.~M. Palhares,
  \href{https://doi.org/10.1109/tits.2014.2341491}{Longitudinal model
  identification and velocity control of an autonomous car}, {IEEE}
  Transactions on Intelligent Transportation Systems (2014) 1--11\href
  {http://dx.doi.org/10.1109/tits.2014.2341491}
  {\path{doi:10.1109/tits.2014.2341491}}.
\newline\urlprefix\url{https://doi.org/10.1109/tits.2014.2341491}

\bibitem{kati2016robust}
M.~S. Kati, H.~K\"{o}ro{\u{g}}lu, J.~Fredriksson,
  \href{https://doi.org/10.1016/j.ifacol.2016.08.046}{Robust lateral control of
  an a-double combination via $\mathcal{H}_{\infty}$ and generalized
  $\mathcal{H}_{2}$ static output feedback}, {IFAC}-{PapersOnLine} 49~(11)
  (2016) 305--311.
\newblock \href {http://dx.doi.org/10.1016/j.ifacol.2016.08.046}
  {\path{doi:10.1016/j.ifacol.2016.08.046}}.
\newline\urlprefix\url{https://doi.org/10.1016/j.ifacol.2016.08.046}

\bibitem{jujnovich2013path}
B.~A. Jujnovich, D.~Cebon,
  \href{https://doi.org/10.1115/1.4023396}{Path-following steering control for
  articulated vehicles}, Journal of Dynamic Systems, Measurement, and Control
  135~(3) (2013) 031006.
\newblock \href {http://dx.doi.org/10.1115/1.4023396}
  {\path{doi:10.1115/1.4023396}}.
\newline\urlprefix\url{https://doi.org/10.1115/1.4023396}

\bibitem{islam2015improve}
M.~M. Islam, L.~Laine, B.~Jacobson,
  \href{https://doi.org/10.1109/itsc.2015.383}{Improve safety by optimal
  steering control of a converter dolly using particle swarm optimization for
  low-speed maneuvers}, in: 2015 {IEEE} 18th International Conference on
  Intelligent Transportation Systems, {IEEE}, 2015, pp. 2370--2377.
\newblock \href {http://dx.doi.org/10.1109/itsc.2015.383}
  {\path{doi:10.1109/itsc.2015.383}}.
\newline\urlprefix\url{https://doi.org/10.1109/itsc.2015.383}

\bibitem{fagnant2015preparing}
D.~J. Fagnant, K.~Kockelman,
  \href{https://doi.org/10.1016/j.tra.2015.04.003}{Preparing a nation for
  autonomous vehicles: opportunities, barriers and policy recommendations},
  Transportation Research Part A: Policy and Practice 77 (2015) 167--181.
\newblock \href {http://dx.doi.org/10.1016/j.tra.2015.04.003}
  {\path{doi:10.1016/j.tra.2015.04.003}}.
\newline\urlprefix\url{https://doi.org/10.1016/j.tra.2015.04.003}

\bibitem{noorvand2017autonomous}
H.~Noorvand, G.~Karnati, B.~S. Underwood,
  \href{https://doi.org/10.3141/2640-03}{Autonomous vehicles}, Transportation
  Research Record: Journal of the Transportation Research Board 2640 (2017)
  21--28.
\newblock \href {http://dx.doi.org/10.3141/2640-03}
  {\path{doi:10.3141/2640-03}}.
\newline\urlprefix\url{https://doi.org/10.3141/2640-03}

\bibitem{Alcala2018}
E.~Alcala, V.~Puig, J.~Quevedo, T.~Escobet, R.~Comasolivas,
  \href{https://doi.org/10.1016/j.conengprac.2017.12.004}{Autonomous vehicle
  control using a kinematic lyapunov-based technique with {LQR}-{LMI} tuning},
  Control Engineering Practice 73 (2018) 1--12.
\newblock \href {http://dx.doi.org/10.1016/j.conengprac.2017.12.004}
  {\path{doi:10.1016/j.conengprac.2017.12.004}}.
\newline\urlprefix\url{https://doi.org/10.1016/j.conengprac.2017.12.004}

\bibitem{Ji2018}
X.~Ji, X.~He, C.~Lv, Y.~Liu, J.~Wu,
  \href{https://doi.org/10.1016/j.conengprac.2018.04.007}{Adaptive-neural-network-based
  robust lateral motion control for autonomous vehicle at driving limits},
  Control Engineering Practice 76 (2018) 41--53.
\newblock \href {http://dx.doi.org/10.1016/j.conengprac.2018.04.007}
  {\path{doi:10.1016/j.conengprac.2018.04.007}}.
\newline\urlprefix\url{https://doi.org/10.1016/j.conengprac.2018.04.007}

\bibitem{Matraji2018}
I.~Matraji, A.~Al-Durra, A.~Haryono, K.~Al-Wahedi, M.~Abou-Khousa,
  \href{https://doi.org/10.1016/j.conengprac.2017.11.009}{Trajectory tracking
  control of skid-steered mobile robot based on adaptive second order sliding
  mode control}, Control Engineering Practice 72 (2018) 167--176.
\newblock \href {http://dx.doi.org/10.1016/j.conengprac.2017.11.009}
  {\path{doi:10.1016/j.conengprac.2017.11.009}}.
\newline\urlprefix\url{https://doi.org/10.1016/j.conengprac.2017.11.009}

\bibitem{Chu2018}
Z.~Chu, Y.~Sun, C.~Wu, N.~Sepehri,
  \href{https://doi.org/10.1016/j.conengprac.2018.02.002}{Active disturbance
  rejection control applied to automated steering for lane keeping in
  autonomous vehicles}, Control Engineering Practice 74 (2018) 13--21.
\newblock \href {http://dx.doi.org/10.1016/j.conengprac.2018.02.002}
  {\path{doi:10.1016/j.conengprac.2018.02.002}}.
\newline\urlprefix\url{https://doi.org/10.1016/j.conengprac.2018.02.002}

\bibitem{Hu2016}
C.~Hu, H.~Jing, R.~Wang, F.~Yan, M.~Chadli,
  \href{https://doi.org/10.1016/j.ymssp.2015.09.017}{Robust
  $\mathcal{H}_{\infty}$ output-feedback control for path following of
  autonomous ground vehicles}, Mechanical Systems and Signal Processing 70-71
  (2016) 414--427.
\newblock \href {http://dx.doi.org/10.1016/j.ymssp.2015.09.017}
  {\path{doi:10.1016/j.ymssp.2015.09.017}}.
\newline\urlprefix\url{https://doi.org/10.1016/j.ymssp.2015.09.017}

\bibitem{Kim2016}
K.-i. Kim, H.~Guan, B.~Wang, R.~Guo, F.~Liang,
  \href{https://doi.org/10.1631/FITEE.1500211}{Active steering control strategy
  for articulated vehicles}, Frontiers of Information Technology {\&}
  Electronic Engineering 17~(6) (2016) 576--586.
\newblock \href {http://dx.doi.org/10.1631/FITEE.1500211}
  {\path{doi:10.1631/FITEE.1500211}}.
\newline\urlprefix\url{https://doi.org/10.1631/FITEE.1500211}

\bibitem{Guan2017}
H.~Guan, K.~Kim, B.~Wang,
  \href{https://doi.org/10.1504/ijhvs.2017.080961}{Comprehensive path and
  attitude control of articulated vehicles for varying vehicle conditions},
  International Journal of Heavy Vehicle Systems 24~(1) (2017) 65.
\newblock \href {http://dx.doi.org/10.1504/ijhvs.2017.080961}
  {\path{doi:10.1504/ijhvs.2017.080961}}.
\newline\urlprefix\url{https://doi.org/10.1504/ijhvs.2017.080961}

\bibitem{Yuan2016}
H.~Yuan, H.~Zhu,
  \href{https://doi.org/10.1080/00423114.2016.1208251}{Anti-jackknife reverse
  tracking control of articulated vehicles in the presence of actuator
  saturation}, Vehicle System Dynamics 54~(10) (2016) 1428--1447.
\newblock \href
  {http://arxiv.org/abs/https://doi.org/10.1080/00423114.2016.1208251}
  {\path{arXiv:https://doi.org/10.1080/00423114.2016.1208251}}, \href
  {http://dx.doi.org/10.1080/00423114.2016.1208251}
  {\path{doi:10.1080/00423114.2016.1208251}}.
\newline\urlprefix\url{https://doi.org/10.1080/00423114.2016.1208251}

\bibitem{Michaek2014}
M.~M. Micha{\l}ek, \href{https://doi.org/10.1016/j.conengprac.2014.04.001}{A
  highly scalable path-following controller for n-trailers with off-axle
  hitching}, Control Engineering Practice 29 (2014) 61--73.
\newblock \href {http://dx.doi.org/10.1016/j.conengprac.2014.04.001}
  {\path{doi:10.1016/j.conengprac.2014.04.001}}.
\newline\urlprefix\url{https://doi.org/10.1016/j.conengprac.2014.04.001}

\bibitem{Nayl2018}
T.~Nayl, G.~Nikolakopoulos, T.~Gustafsson, D.~Kominiak, R.~Nyberg,
  \href{https://doi.org/10.1016/j.robot.2018.01.006}{Design and experimental
  evaluation of a novel sliding mode controller for an articulated vehicle},
  Robotics and Autonomous Systems 103 (2018) 213--221.
\newblock \href {http://dx.doi.org/10.1016/j.robot.2018.01.006}
  {\path{doi:10.1016/j.robot.2018.01.006}}.
\newline\urlprefix\url{https://doi.org/10.1016/j.robot.2018.01.006}

\bibitem{Scherer1995}
C.~Scherer, \href{https://doi.org/10.1007/978-1-4471-3061-1_8}{Mixed h 2/h
  $\infty$ control}, in: Trends in Control, Springer London, 1995, pp.
  173--216.
\newblock \href {http://dx.doi.org/10.1007/978-1-4471-3061-1_8}
  {\path{doi:10.1007/978-1-4471-3061-1_8}}.
\newline\urlprefix\url{https://doi.org/10.1007/978-1-4471-3061-1_8}

\bibitem{Cerri2014}
M.~H. Terra, J.~P. Cerri, J.~Y. Ishihara,
  \href{https://doi.org/10.1109/tac.2014.2309282}{Optimal robust linear
  quadratic regulator for systems subject to uncertainties}, {IEEE}
  Transactions on Automatic Control 59~(9) (2014) 2586--2591.
\newblock \href {http://dx.doi.org/10.1109/tac.2014.2309282}
  {\path{doi:10.1109/tac.2014.2309282}}.
\newline\urlprefix\url{https://doi.org/10.1109/tac.2014.2309282}

\bibitem{cerri2009recursive}
J.~P. Cerri, M.~H. Terra, J.~Y. Ishihara,
  \href{https://doi.org/10.1109/acc.2009.5160553}{Recursive robust regulator
  for discrete-time state-space systems}, in: 2009 American Control Conference,
  {IEEE}, 2009, pp. 3077--3082.
\newblock \href {http://dx.doi.org/10.1109/acc.2009.5160553}
  {\path{doi:10.1109/acc.2009.5160553}}.
\newline\urlprefix\url{https://doi.org/10.1109/acc.2009.5160553}

\bibitem{LI2018}
C.~Li, H.~Jing, R.~Wang, N.~Chen,
  \href{http://www.sciencedirect.com/science/article/pii/S0888327017304892}{Vehicle
  lateral motion regulation under unreliable communication links based on
  robust $\mathcal{H}_{\infty}$ output-feedback control schema}, Mechanical
  Systems and Signal Processing 104 (2018) 171 -- 187.
\newblock \href {http://dx.doi.org/https://doi.org/10.1016/j.ymssp.2017.09.012}
  {\path{doi:https://doi.org/10.1016/j.ymssp.2017.09.012}}.
\newline\urlprefix\url{http://www.sciencedirect.com/science/article/pii/S0888327017304892}

\bibitem{ZHAO2018}
W.~Zhao, H.~Zhang, Y.~Li,
  \href{http://www.sciencedirect.com/science/article/pii/S0888327017306751}{Displacement
  and force coupling control design for automotive active front steering
  system}, Mechanical Systems and Signal Processing 106 (2018) 76 -- 93.
\newblock \href {http://dx.doi.org/https://doi.org/10.1016/j.ymssp.2017.12.037}
  {\path{doi:https://doi.org/10.1016/j.ymssp.2017.12.037}}.
\newline\urlprefix\url{http://www.sciencedirect.com/science/article/pii/S0888327017306751}

\bibitem{skjetne2001nonlinear}
R.~Skjetne, T.~Fossen,
  \href{https://doi.org/10.1109/oceans.2001.968121}{Nonlinear maneuvering and
  control of ships}, in: MTS/IEEE Oceans 2001. An Ocean Odyssey. Conference
  Proceedings (IEEE Cat. No.01CH37295), Vol.~3, Marine Technol. Soc, 2001, pp.
  1808--1815.
\newblock \href {http://dx.doi.org/10.1109/oceans.2001.968121}
  {\path{doi:10.1109/oceans.2001.968121}}.
\newline\urlprefix\url{https://doi.org/10.1109/oceans.2001.968121}

\bibitem{schramm2014vehicle}
D.~Schramm, M.~Hiller, R.~Bardini,
  \href{https://doi.org/10.1007/978-3-662-54483-9}{Vehicle Dynamics}, Springer
  Berlin Heidelberg, 2018.
\newblock \href {http://dx.doi.org/10.1007/978-3-662-54483-9}
  {\path{doi:10.1007/978-3-662-54483-9}}.
\newline\urlprefix\url{https://doi.org/10.1007/978-3-662-54483-9}

\bibitem{van2012analysis}
M.~F. van~de Molengraft-Luijten, I.~J. Besselink, R.~M. Verschuren,
  H.~Nijmeijer, \href{https://doi.org/10.1080/00423114.2012.676650}{Analysis of
  the lateral dynamic behaviour of articulated commercial vehicles}, Vehicle
  System Dynamics 50~(sup1) (2012) 169--189.
\newblock \href {http://dx.doi.org/10.1080/00423114.2012.676650}
  {\path{doi:10.1080/00423114.2012.676650}}.
\newline\urlprefix\url{https://doi.org/10.1080/00423114.2012.676650}

\bibitem{fancher1989directional}
P.~S. Fancher, \href{https://doi.org/10.4271/892499}{Directional dynamics
  considerations for multi-articulated, multi-axled heavy vehicles}, in: {SAE}
  Technical Paper Series, {SAE} International, 1989, pp. 630--640.
\newblock \href {http://dx.doi.org/10.4271/892499} {\path{doi:10.4271/892499}}.
\newline\urlprefix\url{https://doi.org/10.4271/892499}

\bibitem{luijten2010lateral}
M.~Luijten, Lateral dynamic behaviour of articulated commercial vehicles,
  Eindhoven University of Technology.

\bibitem{houben2008analysis}
L.~Houben, Analysis of truck steering behaviour using a multi-body model,
  Master's thesis, Eindhoven University of Technology, DCT 343.

\bibitem{kailath2000linear}
T.~Kailath, A.~H. Sayed, B.~Hassibi, Linear Estimation, Prentice Hall, New
  Jersey, 2000.

\bibitem{sayed1999design}
A.~H. Sayed, V.~H. Nascimento, \href{https://doi.org/10.1007/bfb0109867}{Design
  criteria for uncertain models with structured and unstructured
  uncertainties}, in: Robustness in identification and control, Springer
  London, 1999, pp. 159--173.
\newblock \href {http://dx.doi.org/10.1007/bfb0109867}
  {\path{doi:10.1007/bfb0109867}}.
\newline\urlprefix\url{https://doi.org/10.1007/bfb0109867}

\bibitem{Sayed2001}
A.~Sayed, \href{https://doi.org/10.1109/9.935054}{A framework for state-space
  estimation with uncertain models}, {IEEE} Transactions on Automatic Control
  46~(7) (2001) 998--1013.
\newblock \href {http://dx.doi.org/10.1109/9.935054}
  {\path{doi:10.1109/9.935054}}.
\newline\urlprefix\url{https://doi.org/10.1109/9.935054}

\bibitem{hassibi1999indefinite}
B.~Hassibi, A.~H. Sayed, T.~Kailath,
  \href{https://doi.org/10.1137/1.9781611970760}{Indefinite-Quadratic
  Estimation and Control}, Society for Industrial and Applied Mathematics,
  1999.
\newblock \href {http://dx.doi.org/10.1137/1.9781611970760}
  {\path{doi:10.1137/1.9781611970760}}.
\newline\urlprefix\url{https://doi.org/10.1137/1.9781611970760}

\end{thebibliography}

\appendix 

\section{Uncertainties matrices}
\label{sec:uncertainties}
In order to estimate the uncertainties matrices $E_{F}$, $E_{G}$ and $H$ in (\ref{eq:RLQRuncertainties}), we considered the inertia uncertainties for deriving a robust control strategy by adopting maximum and minimum values of payload, which results in  $m_{p_{max}}$ and $m_{p_{min}}$, respectively. Setting  these values of mass, the maximum variations of the matrices $F$ and $G$ are calculated as:
\begin{equation}
\Gamma_{F} = F_{m_{p_{min}}}-F_{m_{p_{max}}}
\end{equation}
where $F_{m_{p_{min}}}$ and $F_{m_{p_{max}}}$ are the discretised state-space matrices when $m_{p_{max}}$ and $m_{p_{min}}$ are applied to the state-space system (\ref{state-space}). Thus, we afterwards selected the row in $\Gamma_{F}$ that is most affected by mass variations, here the fifth row.

In this work, $m_{p_{min}}$ and $m_{p_{min}}$ values correspond to the unloaded and $100\%$ of overload vehicle operation. Consequently, this range of uncertainty may imply in large $E_F$ and $E_G$ values, denoted as $E_{F_{100\%}}$ and $E_{G_{100\%}}$. The selection of these values during the control design guarantees the robustness stability for any condition, satisfying the mass variability. However, it jeopardises the performance of the nominal case. Therefore, lower $E_{F}$ and $E_{G}$ values were considered to overcome this problem. This choice is capable of enhancing the robustness of the proposal without jeopardising the system performance.

Thus, the matrices $E_{F_i}$ and $E_{G_i}$ are obtained as follows:

%--
\begin{equation}
E_{F_{i}} = 
\begin{bmatrix}
1\\
1\\
1\\
1\\
1\\
0.1\\
\end{bmatrix}^T
\begin{bmatrix}
\underset{\Gamma_{F_{5,1}}}{arg}(|\Gamma_{F_{5,1}}|) & 0 & 0 & 0 & 0 & 0\\ 0 & \underset{\Gamma_{F_{5,2}}}{arg}(|\Gamma_{F_{5,2}}|) & 0 & 0 & 0 & 0\\ 0 & 0 & \underset{\Gamma_{F_{5,3}}}{arg}(|\Gamma_{F_{5,3}}|) & 0 & 0 & 0\\ 0 & 0 & 0 & \underset{\Gamma_{F_{5,4}}}{arg}(|\Gamma_{F_{5,4}}|) & 0 & 0\\ 0 & 0 & 0 & 0 & \underset{\Gamma_{F_{5,5}}}{arg}(|\Gamma_{F_{5,5}}|) & 0\\ 0 & 0 & 0 & 0 & 0 & \underset{\Gamma_{F_{5,6}}}{arg}(|\Gamma_{F_{5,6}}|)
\end{bmatrix}
\end{equation}
%--
\begin{equation}
\label{Eg_calc}
E_{G_{i}} = 
\begin{bmatrix}
0.1 \\ 
0.1
\end{bmatrix}^T
\begin{bmatrix}
\underset{\Gamma_{F_{5,j}}}{arg} \{max(|\Gamma_{F_{5,j}}|)\} & 0\\
0 & \underset{\Gamma_{F_{5,j}}}{arg} \{max(|\Gamma_{F_{5,j}}|)\}
\end{bmatrix}
\end{equation}
%--
and $H_{i} = [1,1,1,1,1,1]^T$. This way, the uncertainties matrices are obtained through \autoref{eq:RLQRuncertainties}
%--
\begin{equation*}
\begin{bmatrix}
\delta {F}_{i} & \delta G_{i}
\end{bmatrix} = H_{i} \Delta_{i} \begin{bmatrix}
E_{F_{i}} & E_{G_{i}}
\end{bmatrix},
\end{equation*}
%--
where $\Delta_{i}$ is a scalar represented by the mass variation.

\section{\texorpdfstring{$\mathcal{H}_{\infty}$}{Lg} control}
\label{h_infty_app}

The robust control design considering $\mathcal{H}_{\infty}$ method mentioned in \cite{hassibi1999indefinite} is used for the following linear system
\begin{equation}
\label{h_infinity_system}
x_{i+1} = F_{i}x_{i}+G_{1,i}w_{i}+G_{2,i}u_{i}, \hspace{2mm} i=0,...,N,
\end{equation}
where $x_{i}$ is the state vector, $u_{i}$ is the control input and $w_{i}$ is the disturbance. In its sub-optimal formulation, this technique is based on finding a control strategy where for every $x_{0}$ and $\{w_{i}\}^{N}_{i=0}$,
%--
\begin{equation}
\frac{x^{*T}_{N+1}P^{c}_{N+1}x^{*}_{N+1}+\sum^{N}_{i=0}(u^{*T}_{i}Q^{c}_{i}u^{*}_{i}+x^{*T}_{i}R^{c}_{i}x^{*}_{i})}{x^{*T}_{0}\prod^{-1}_{0}x^{*}_{0}+\sum^{N}_{i=0}(w^{*T}_{i}Q^{w}_{i}w^{*}_{i})}<\gamma^{2},
\end{equation}
%--
for a suitable $\gamma>0$, where $P^{c}_{N+1}$, $Q^{c}_{i}$, $R^{c}_{i}$, $\Pi_{0}$ and $Q^{w}_{i}$  are non-negative definite weighing matrices. Such matrices are associated with the final state, control input, state, initial state and disturbance, respectively. The recursive solution of this problem is formulated in terms of backwards Riccati equation and the verification of some existence conditions is necessary.

To perform the control using this technique, uncertain system \eqref{eq:statespaceDT}-\eqref{eq:RLQRuncertainties} must be rewritten as the system \eqref{h_infinity_system}. Hence, the following immediate identifications are considered
%--
\begin{equation}
\begin{aligned}
\label{h_infinity_identification}
F_{i}\leftarrow F_{i},\hspace{2mm} G_{2,i}\leftarrow G_{i} \hspace{2mm} x_{i}\leftarrow x_{i}, \hspace{2mm} &u_{i}\leftarrow u_{i}, \hspace{2mm} G_{1,i} \leftarrow H_{i}, \hspace{2mm} w_{i}\leftarrow \Delta_{i}\begin{bmatrix} E_{F_{i}} & E_{G_{i}} \end{bmatrix}\begin{bmatrix} x_{i}\\u_{i} \end{bmatrix},\\
P^{c}_{N+1}\leftarrow P_{N+1}, \hspace{2mm} Q^{c}_{i} \leftarrow &R_{i}, \hspace{2mm} R^{c}_{i} \leftarrow Q_{i}, \hspace{2mm} Q^{w}_{i} \leftarrow I, \hspace{2mm} \Pi_{0}\leftarrow I.
\end{aligned}
\end{equation}
%--

In order to use the system \eqref{h_infinity_system} as the uncertain system \eqref{eq:statespaceDT}-\eqref{eq:RLQRuncertainties}, some algebraic manipulations are necessary. Considering ${R_{G,i}^{c}}^{-1} = Q_{i}^{c}+G_{2,i}^{T}P^{c}_{N+1}G_{2,i}$, the control signal is obtained from \cite{hassibi1999indefinite} as 
%--
\begin{equation}
\label{u1}
\begin{aligned}
   u_{i} &= -{R_{G,i}^{c}}^{-1}G_{2,i}^{T}P^{c}_{N+1}F_{i}x_{i} -{R_{G,i}^{c}}^{-1}G_{2,i}^{T}P^{c}_{N+1}G_{1,i} w_{i}\\
        &=(-{R_{G,i}^{c}}^{-1}G_{2,i}^{T}P^{c}_{N+1})(F_{i}x_{i}+G_{1,i}w_{i}).
\end{aligned}
\end{equation}
%--    
From identifications made in \eqref{h_infinity_identification} substitutions can be made in $G_{1,i}$ and $w_{i}$
\begin{equation*}
\begin{aligned}
   u_{i} &= (-{R_{G,i}^{c}}^{-1}G_{2,i}^{T}P^{c}_{N+1})((F_{i}+H_{i}\Delta E_{F_{i}})x_{i}+H_{i}\Delta E_{G_{i}}u_{i-1})\\
        &= -{R_{G,i}^{c}}^{-1}G_{2,i}^{T}P^{c}_{N+1} (F_{i}+\delta F_{i})x_{i} -{R_{G,i}^{c}}^{-1}G_{2,i}^{T}P^{c}_{N+1}H_{i}\Delta E_{G_{i}}u_{i-1},
\end{aligned}
\end{equation*}
and adding $G_{2}$ in both sides of the equation the \eqref{u1} becomes
\begin{equation}
\label{u2}
\begin{aligned}
   u_{i} &= -{R_{G,i}^{c}}^{-1}G_{2,i}^{T}P^{c}_{N+1} (F_{i}+\delta F_{i})x_{i}
    -{R_{G,i}^{c}}^{-1}G_{2,i}^{T}P^{c}_{N+1}(G_{2,i}+H_{i}\Delta E_{G_{i}})u_{i-1} 
    + {R_{G,i}^{c}}^{-1}G_{2,i}^{T}P^{c}_{N+1}G_{2,i}u_{i-1}\\
        &= -{R_{G,i}^{c}}^{-1}G_{2,i}^{T}P^{c}_{N+1} (F_{i}+\delta F_{i})x_{i}
    -{R_{G,i}^{c}}^{-1}G_{2,i}^{T}P^{c}_{N+1}(G_{2,i}+\delta G_{2,i})u_{i-1} 
    + {R_{G,i}^{c}}^{-1}G_{2,i}^{T}P^{c}_{N+1}G_{2,i}u_{i-1}.
\end{aligned}
\end{equation}
Considering the sampling period sufficiently small so that $x_{i}\approx x_{i-1}$, the \eqref{u2} can be rewritten as 
\begin{equation*}
    u_{i} \coloneqq -{R_{G,i}^{c}}^{-1}G_{2,i}^{T}P^{c}_{N+1}((F_{i}+\delta F_{i})x_{i-1}+(G_{2,i}+\delta G_{2,i})u_{i-1}) +{R_{G}^{c}}^{-1}G_{2}^{T}P^{c}_{N+1}G_{2,i}u_{i-1}.
\end{equation*}
Hence, as $u_{i-1}=z_{i}$, the system \eqref{h_infinity_system} can be rewritten as system \eqref{eq:statespaceDT}-\eqref{eq:RLQRuncertainties} 
\begin{gather*}
\begin{aligned}
x_{i+1} = (F_{i} + \delta F_{i})x_{i} + (G_{2,i}  + \delta G_{2,i})u_{i}\\
\begin{bmatrix}
\delta {F}_{i} & \delta G_{i}
\end{bmatrix} 
= H_{i} \Delta_{i} \begin{bmatrix}
E_{F_{i}} & E_{G_{i}}
\end{bmatrix},
\end{aligned}    
\end{gather*}
where 
\begin{equation}
    u_{i} = -{R_{G,i}^{c}}^{-1}G_{2,i}^{T}P^{c}_{N+1}x_{i} +{R_{G,i}^{c}}^{-1}G_{2,i}^{T}P^{c}_{N+1}G_{2,i}z_{i}.
\end{equation}

See details in \cite{hassibi1999indefinite}.
\end{document}